\documentclass[final,5p,times,twocolumn]{elsarticle}

\usepackage{amsmath,amssymb,amsfonts,amsthm}
\newtheorem{definition}{Definition}
\newtheorem{theorem}{Theorem}
\usepackage{algorithmic}
\usepackage{algorithm}
\usepackage{graphicx}
\usepackage{textcomp}
\usepackage{xcolor}
\usepackage{booktabs}
\usepackage{url}
\usepackage{multirow}
\usepackage{bm}

\newcommand{\R}{\mathbb{R}}
\newcommand{\safeSet}{\Omega_{\mathrm{safe}}}
\newcommand{\unsafeSet}{\Omega_{\mathrm{unsafe}}}
\newcommand{\constraint}{h}

\journal{Preprint}

\usepackage{hyperref}

\begin{document}

\begin{frontmatter}

\title{Hamilton-Jacobi Reachability Based Safe Reinforcement Learning for Emergency Collision Avoidance}

\author[inst1]{Yuhong Jiang\fnref{fn1}}\ead{jiang-yh23@mails.tsinghua.edu.cn}
\author[inst1]{Shiyue Zhao\fnref{fn1}}\ead{zhaosy21@mails.tsinghua.edu.cn}
\author[inst1,inst2]{Junzhi Zhang\corref{cor1}}\ead{jzhzhang@tsinghua.edu.cn}
\author[inst1]{Junfeng Zhang}\ead{zhangjf20@mails.tsinghua.edu.cn}
\author[inst1]{Xinhan Li}\ead{lixinhan21@mails.tsinghua.edu.cn}
\author[inst1]{Shijie Zhao}\ead{zhaosj17@mails.jlu.edu.cn}
\author[inst1,inst2]{Chengkun He}\ead{hechengkun@tsinghua.edu.cn}

\cortext[cor1]{Corresponding author.}
\fntext[fn1]{These authors contributed equally to this work.}

\affiliation[inst1]{organization={School of Vehicle and Mobility, Tsinghua University}, city={Beijing}, country={China}}
\affiliation[inst2]{organization={State Key Laboratory of Intelligent Green Vehicle and Mobility, Tsinghua University}, city={Beijing}, country={China}}

\begin{abstract}
Emergency collision avoidance under extreme driving conditions demands safety-critical control that accounts for both obstacle proximity and vehicle dynamic stability over a future time horizon, yet existing methods often rely on instantaneous or local safety evaluations.
This paper proposes a safe reinforcement learning framework guided by a Hamilton-Jacobi (HJ) reachability based motion safety set that provides forward-looking safety supervision for constrained policy optimization.
Specifically, a unified signed safety function is formulated by combining geometric collision margins and chassis stability limits, and is then extended through reachability analysis into a finite-horizon motion safety set that characterizes whether safety can be maintained under future vehicle state evolution.
To enable practical computation, the motion safety set is approximated from offline extreme driving data, mitigating the computational burden of grid-based HJ solvers.
The learned motion safety set is then embedded as a continuous safety cost into a constrained Markov decision process, and a PID-Lagrangian policy optimization scheme is employed to adaptively regulate the Lagrange multiplier for safety constraint enforcement.
Simulation and real-vehicle experiments on low-adhesion obstacle-avoidance scenarios demonstrate that the proposed method achieves higher goal-reaching rates, produces smoother avoidance maneuvers, and maintains larger unified safety margins than baseline methods.
\end{abstract}


\end{frontmatter}

\section{Introduction}
\label{sec:introduction}

The automotive industry has made significant progress in vehicle intelligence, enabling autonomous vehicles to operate with increasing autonomy in complex traffic environments.
However, ensuring reliable behavior in safety-critical edge cases remains a fundamental challenge \cite{ahangarnejad2021review}.
Among various safety-critical functions, emergency collision avoidance is particularly demanding: when a sudden obstacle appears in the vehicle's path, the system must execute a rapid evasive maneuver within a very short time window, while the available road space, tire-road friction, and vehicle dynamic capability may all be severely limited \cite{tan2025autonomous}.
Unlike nominal driving scenarios, where the vehicle usually operates with sufficient stability margins, emergency maneuvers often push the vehicle close to the limits of tire adhesion and chassis stability.
In this regime, collision avoidance and vehicle stability are no longer separable requirements: a maneuver that is geometrically collision-free may still induce excessive sideslip, yaw instability, or loss of control, whereas an overly conservative stability-oriented action may fail to avoid the obstacle in time \cite{zhao2022beyond,wang2023integrated}.
Therefore, emergency collision avoidance requires a safety representation that can jointly describe external obstacle proximity and internal vehicle dynamic feasibility over future state evolution.
This coupling calls for a systematic characterization of the vehicle's forward-looking motion safety boundary, as well as an effective method to leverage this boundary for improving emergency collision-avoidance safety.

\begin{figure*}[t]
\centering
\includegraphics[width=\textwidth]{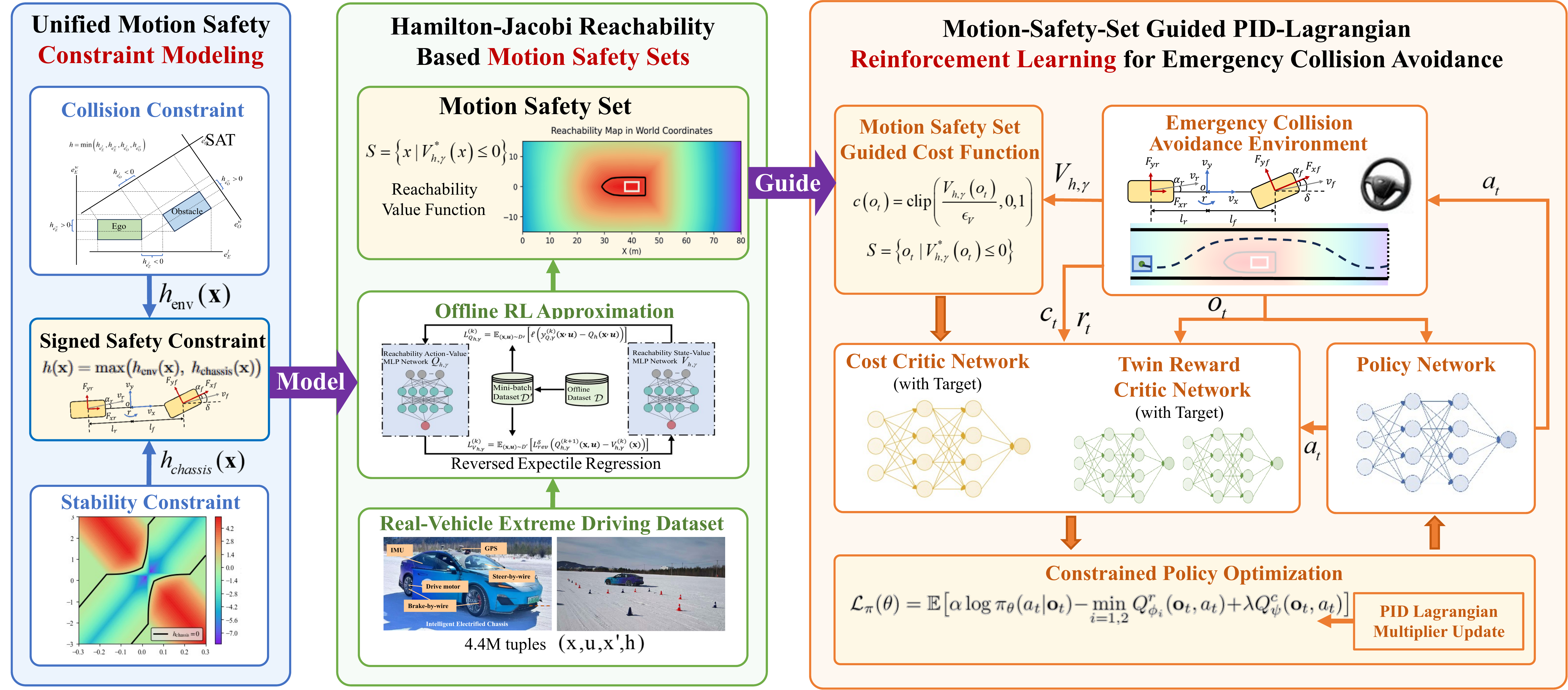}
\caption{Overview of the proposed Hamilton-Jacobi reachability based safe reinforcement learning for emergency collision avoidance.}
\label{fig:framework}
\end{figure*}

Existing safety methods for emergency driving can be broadly viewed as attempts to evaluate, approximate, or enforce a motion safety boundary. 
Early approaches relied on rule-based and geometry-based safety criteria, such as safety-distance constraints and time-to-collision  thresholds, to evaluate whether the ego vehicle is sufficiently separated from surrounding obstacles \cite{gipps1981behavioural,vogel2003comparison}.
Subsequent methods introduced artificial potential fields \cite{rasekhipour2017potential} and driving risk fields \cite{wang2016driving} to represent surrounding danger in a continuous form, while control barrier functions \cite{ames2017control} and predictive safety filters \cite{wabersich2021predictive} further moved from passive risk evaluation to active safety enforcement by restricting admissible control inputs or modifying unsafe commands during execution.

Although these methods have improved vehicle safety, most rely on instantaneous criteria that may fail under aggressive maneuvers, where a currently safe state can still become unsafe under future inputs and disturbances \cite{hsu2024safety}.
A reliable motion safety boundary should therefore account for future dynamic evolution rather than only the current state.
Hamilton-Jacobi reachability theory provides a natural framework for constructing such a dynamics-aware safety boundary \cite{mitchell2005time,margellos2011hamilton}.
By reasoning over system dynamics, control authority, and disturbances, HJ reachability characterizes whether a state can remain safe over a future horizon, making it well suited for emergency driving where collision avoidance and dynamic feasibility must be evaluated together.
Nevertheless, classical HJ solvers rely on grid-based dynamic programming and suffer from the curse of dimensionality, making them intractable for the high-dimensional vehicle dynamics needed to capture realistic extreme driving behavior \cite{Leung2020ReachabilitySafetyAssurance,bansal2021deepreach}.
Recent advances in learning-based reachability approximation have shown that the reachability value function can be estimated directly from data using reinforcement learning and offline value-function learning \cite{fisac2019bridging,herbert2021scalable,hsu2021safety,zheng2024safe}.
However, how to transform learning-based reachability into a practical motion-safety representation that is tightly coupled with vehicle dynamic capability for emergency collision avoidance remains insufficiently explored.

Beyond safety boundary construction, it remains to effectively incorporate such a boundary into the collision-avoidance decision-making and control loop.
Optimization-based controllers, especially model predictive control (MPC) \cite{shang2023emergency,zhou2024personalized} and predictive safety-filter methods \cite{tearle2021predictive}, are able to incorporate safety constraints during online decision making and trajectory tracking.
However, their performance under extreme maneuvers is often affected by model mismatch, nonlinear system dynamics, constraint feasibility, and real-time computational requirements \cite{goh2024beyond}.
Reinforcement learning (RL) provides a complementary approach that directly learns nonlinear collision-avoidance policies from interaction data, and has demonstrated strong performance in complex driving tasks \cite{kiran2022deep,aradi2022survey}.
However, standard RL usually encodes safety through reward shaping, which is a soft constraint that cannot guarantee safety satisfaction and whose effectiveness depends heavily on reward design.
To address this limitation, constrained RL formulates safety-critical control as a constrained Markov decision process (CMDP), where the task objective and safety requirement are separated into reward and cost signals, and the policy is optimized under prescribed safety constraints \cite{achiam2017constrained}.
Recent studies have applied constrained RL to the decision-making and control architecture of autonomous vehicles, demonstrating enhanced safety performance in various driving scenarios \cite{nguyen2023safe,luo2025stability,he2023robust}.
However, the safety constraints in these methods are still often constructed from task-specific indicators, lacking consideration of state evolution safety under vehicle kinematic and dynamic constraints.
This gap motivates leveraging a dynamics-aware motion safety cost into constrained RL to improve the safety performance of autonomous vehicles.

Considering these issues, this paper proposes a safe reinforcement learning framework for emergency collision avoidance, in which a Hamilton-Jacobi reachability based motion safety set is constructed and used to guide constrained policy optimization.
The main contributions of this paper are summarized as follows:
\begin{itemize}
\item \textbf{A unified signed safety function}, integrating geometric collision margins and lateral-yaw stability limits into a single safety interface for emergency collision avoidance.
\item \textbf{A Hamilton-Jacobi reachability based motion safety set}, which lifts the instantaneous safety function into a finite-horizon safety representation through discounted reachability Bellman operator, with the corresponding reachability value approximated from offline extreme-driving data.
\item \textbf{A motion-safety-set-guided constrained reinforcement learning method}, which converts the learned reachability value function into a continuous safety cost for constrained policy optimization. The resulting constrained problem is solved using PID-Lagrangian SAC to improve safety-constraint regulation during learning.
\item \textbf{Validation through simulation and real-vehicle experiments}, demonstrating improved goal-reaching performance, smoother avoidance behavior, and larger motion safety margins in low-adhesion obstacle-avoidance scenarios.
\end{itemize}

\section{Framework Overview}
\label{sec:framework}

The overall framework of the proposed method is shown in Fig.~\ref{fig:framework}.
It follows a modeling--reachability--control pipeline that connects three levels of safety reasoning: unified constraint modeling, motion safety set approximation based on Hamilton-Jacobi reachability, and constrained policy optimization that integrates the learned safety set into reinforcement learning.

First, a unified motion safety constraint is constructed to provide a common safety interface for both external collision avoidance and internal chassis stability.
The external environment constraint evaluates the geometric collision-avoidance margin between the ego vehicle and surrounding obstacles based on oriented-rectangle geometry.
The internal chassis constraint characterizes the lateral-yaw stability limit of the vehicle under near-limit maneuvers.
These two constraints are combined into a single signed safety function, which serves as the basis for the subsequent reachability analysis.

Second, the instantaneous signed safety constraint is lifted to a dynamic motion safety set through Hamilton-Jacobi reachability.
This step transforms current-state safety evaluation into a forward-looking safety representation that accounts for future vehicle-state evolution and disturbances.
To avoid the computational burden of grid-based HJ solvers in high-dimensional vehicle dynamics, the reachability value function is approximated from real-vehicle extreme driving data using offline reinforcement learning.

Third, the learned motion safety set is incorporated into safe reinforcement learning for emergency collision avoidance.
The reachability value function is used to construct the safety cost in a constrained Markov decision process.
A PID-Lagrangian SAC algorithm is then used to optimize the control policy, so that the resulting collision-avoidance maneuver remains within the learned motion safety boundary.

The following sections present the details of each module.
Section~\ref{sec:modeling} establishes the vehicle dynamics model and the unified motion safety constraint.
Section~\ref{sec:hj_reachability} formulates the reachability-based motion safety set and its offline approximation.
Section~\ref{sec:safe_rl_control} introduces the motion-safety-set guided safe reinforcement learning algorithm.
Section~\ref{sec:validation} further validates the proposed framework through simulation and real-vehicle experiments.

\section{Unified Modeling of Vehicle Dynamics and Safety Constraints Under Extreme Driving Conditions}
\label{sec:modeling}

\subsection{Vehicle Dynamics Modeling}

\begin{figure}[t]
\centering
\includegraphics[width=0.65\columnwidth]{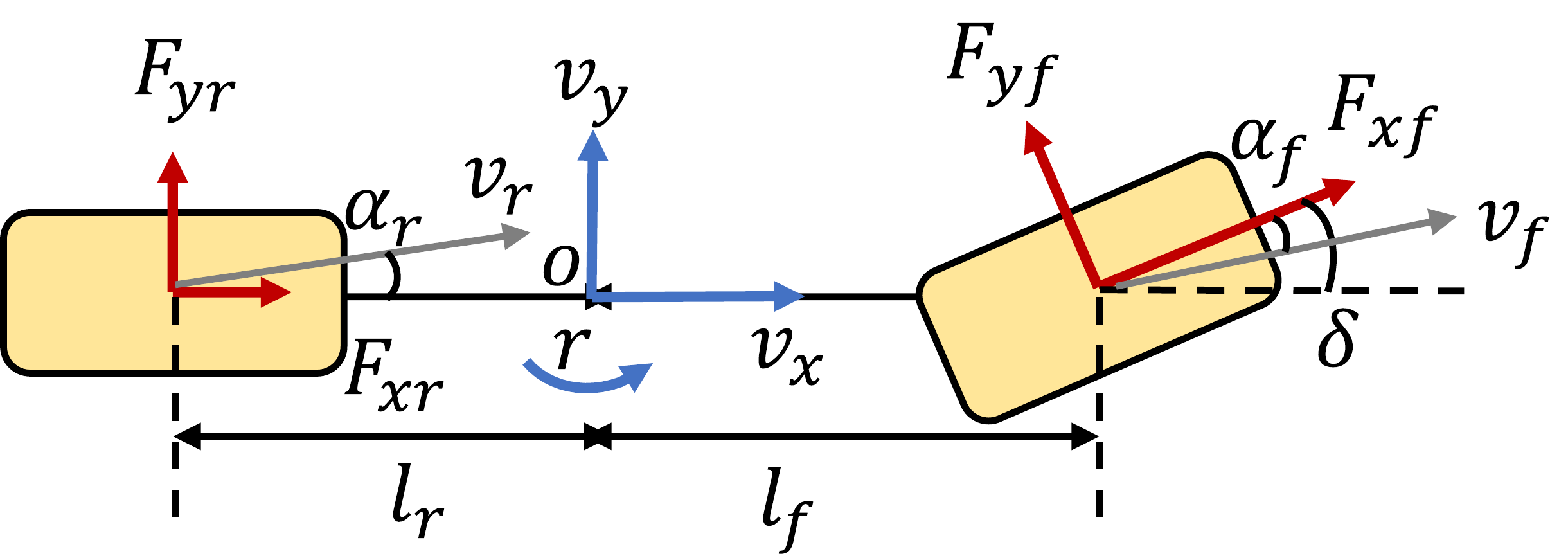}
\caption{Nonlinear vehicle dynamics model.}
\label{fig:vehicle_dynamics}
\end{figure}

To accurately capture the complex dynamics of a vehicle under extreme driving conditions,
we adopt a nonlinear vehicle dynamics model as shown in Fig.~\ref{fig:vehicle_dynamics}.
The vehicle is modeled as a 3-DOF bicycle model with the Fiala tire model.
The longitudinal, lateral translation and yaw motion of the vehicle body can be expressed as:
\begin{align}
m(\dot{v}_x - v_y r) &= F_{xf} \cos\delta_f - F_{yf} \sin\delta_f + F_{xr} - F_{\text{aero}} \label{eq:longitudinal} \\
m(\dot{v}_y + v_x r) &= F_{yf} \cos\delta_f + F_{xf} \sin\delta_f + F_{yr} \label{eq:lateral} \\
I_z \dot{r} &= l_f (F_{yf} \cos\delta_f + F_{xf} \sin\delta_f) - l_r F_{yr} \label{eq:yaw_3dof}
\end{align}
where $m$ is the total vehicle mass, $I_z$ is the yaw moment of inertia, $v_x$ and $v_y$ are the longitudinal and lateral velocities, $r$ is the yaw rate, $\delta_f$ is the front steering angle, $l_f$ and $l_r$ are the distances from the center of mass to the front and rear axles, $F_{xf}$ and $F_{xr}$ are the longitudinal tire forces, $F_{yf}$ and $F_{yr}$ are the lateral tire forces, and $F_{\text{aero}}$ is the aerodynamic drag force.

The front and rear tire slip angles are defined as:
\begin{align}
\alpha_f &= \delta_f - \arctan\left(\frac{v_y + l_f r}{v_x}\right) \label{eq:alpha_f} \\
\alpha_r &= -\arctan\left(\frac{v_y - l_r r}{v_x}\right) \label{eq:alpha_r}
\end{align}

During driving maneuvers, the vertical load transfers between the front and rear axles due to inertial forces.
The front and rear vertical loads are computed as:
\begin{align}
F_{zf} &= \frac{m g l_r - m_s (\dot{v}_x - v_y r) h_{\text{cg}}}{l_f + l_r} \label{eq:Fzf} \\
F_{zr} &= \frac{m g l_f + m_s (\dot{v}_x - v_y r) h_{\text{cg}}}{l_f + l_r} \label{eq:Fzr}
\end{align}
where $m_s$ is the sprung mass, $g$ is the gravitational acceleration, and $h_{\text{cg}}$ is the height of the center of gravity.

The tire forces are modeled using the Fiala tire model, which captures the nonlinear transition from elastic deformation to full sliding saturation.
Defining the critical slip angle $\alpha_{sl} = \arctan(3F_y^{\max}/C_\alpha)$ as the boundary between regimes, the lateral force is:
\begin{subequations}
\label{eq:fiala_3dof}
\begin{align}
F_y &= -C_\alpha \tan\alpha + \frac{C_\alpha^2}{3F_y^{\max}} |\tan\alpha| \tan\alpha \notag \\
&\quad - \frac{C_\alpha^3}{27(F_y^{\max})^2} \tan^3\alpha, \quad |\alpha| \leq \alpha_{sl} \label{eq:fiala_elastic} \\
F_y &= -F_y^{\max}\operatorname{sgn}(\alpha), \quad |\alpha| > \alpha_{sl} \label{eq:fiala_saturated}
\end{align}
\end{subequations}
where $C_\alpha$ is the cornering stiffness, and the maximum lateral force is limited by the friction circle constraint $F_y^{\max} = \sqrt{(\mu F_z)^2 - F_x^2}$ with road adhesion coefficient $\mu$, vertical load $F_z$, and longitudinal force $F_x$.

\subsection{External Environment Constraints Modeling for Collision Avoidance}
\label{sec:geometric_constraint}

Under emergency collision-avoidance maneuvers, the environmental safety constraint is mainly determined by the relative geometry between the ego vehicle and surrounding obstacles.
Circular or point-mass approximations may distort the true occupied region of road vehicles, especially in aggressive driving scenarios,
causing the collision constraints to be overly conservative or insufficiently accurate for high-performance decision-making and control.
Therefore, both the ego vehicle and the obstacle are modeled as oriented rectangles.
As illustrated in Fig.~\ref{fig:sat}, the separating axis theorem (SAT) is applied to convert the binary collision test into a continuous signed geometric safety margin.

\begin{figure}[t]
\centering
\includegraphics[width=0.8\columnwidth]{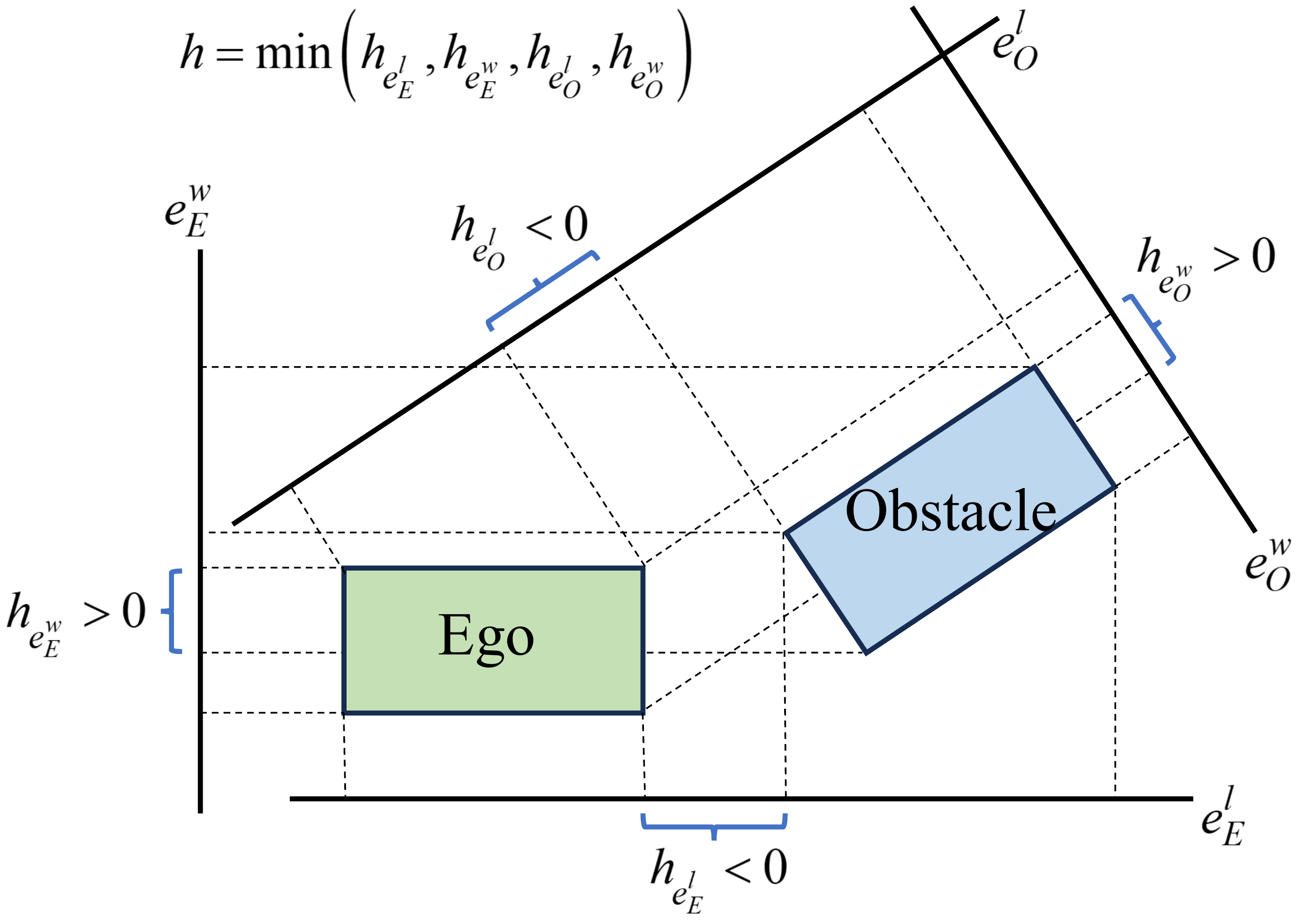}
\caption{Illustration of the SAT-based collision detection and signed safety margin.}
\label{fig:sat}
\end{figure}

For each candidate axis $a\in\mathcal{A}=\{e_E^l,e_E^w,e_O^l,e_O^w\}$, the signed overlap margin between the ego vehicle and the obstacle is defined as $h_a(\mathbf{x}) = \rho_E(a)+\rho_O(a) - |a^\top(c_E-c_O)|$ (see Appendix~\ref{app:sat} for detailed derivation).
Since collision-free separation is certified by the existence of at least one separating axis, the global environmental safety function is:
\begin{equation}
h_{\mathrm{env}}(\mathbf{x}) = \min_{a\in\mathcal{A}} h_a(\mathbf{x}).
\label{eq:henv}
\end{equation}

This function is continuous, piecewise smooth, and locally Lipschitz, satisfying $h_{\mathrm{env}}(\mathbf{x}) < 0$ when collision-free, $h_{\mathrm{env}}(\mathbf{x}) > 0$ when in collision, and $h_{\mathrm{env}}(\mathbf{x}) = 0$ on the contact boundary.

\subsection{Internal Vehicle Stability Constraints Analysis and Modeling}

Under extreme driving conditions, the vehicle may approach its physical limits in terms of both lateral responsiveness and yaw stability.
In such cases, traditional linear assumptions become invalid, and the tire dynamics enter a highly nonlinear regime due to saturation effects and load transfer.
To model and constrain the vehicle's safe handling region, we define a lateral-yaw stability envelope in the $\beta$--$r$ phase plane.
Inspired by \cite{li2022integrated}, the chassis stability envelope is constructed by combining a linearized tire dynamics model with phase-plane partitioning in the $\beta$--$r$ space.

Under unsaturated conditions, the Fiala model in \eqref{eq:fiala_3dof} can be linearized, yielding the equivalent cornering stiffness $\bar{C}_\alpha = \partial F_y / \partial \alpha|_{\alpha}$ that characterizes the tire's instantaneous responsiveness to steering inputs.
By assuming constant $v_x$ and reducing the 3-DOF model to a 2-DOF lateral-yaw subsystem, the yaw-sideslip dynamics can be expressed as:
\begin{equation}
\label{eq:state_space}
\begin{bmatrix} \dot{\beta} \\ \dot{r} \end{bmatrix} = \mathbf{A} \begin{bmatrix} \beta \\ r \end{bmatrix} + \mathbf{B} \delta_f
\end{equation}
where the system matrix $\mathbf{A}$ depends on the equivalent cornering stiffnesses $\bar{C}_{\alpha f}$, $\bar{C}_{\alpha r}$, vehicle parameters, and longitudinal speed (see Appendix~\ref{app:stability} for the full expression).

Applying the Routh-Hurwitz criterion to \eqref{eq:state_space}, the stability condition reduces to:
\begin{equation}
\label{eq:stability_boundary}
\bar{C}_{\alpha f} \bar{C}_{\alpha r} (l_f + l_r)^2 + m v_x^2 (\bar{C}_{\alpha r} l_r - \bar{C}_{\alpha f} l_f) > 0
\end{equation}
which partitions the $(\beta\text{-}r)$ phase plane into stable and unstable regions, as illustrated in Fig.~\ref{fig:stability_envelope}.

\begin{figure}[t]
\centering
\includegraphics[width=0.9\columnwidth]{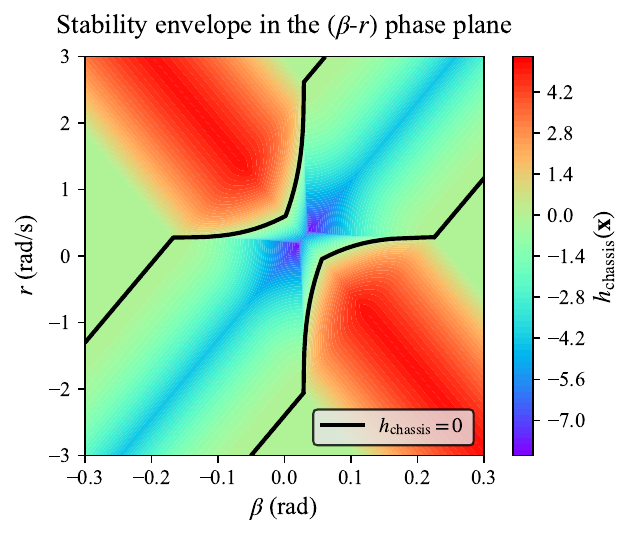}
\caption{Stability envelope in the $(\beta\text{-}r)$ phase plane at a certain state.}
\label{fig:stability_envelope}
\end{figure}

Based on this condition, a signed chassis safety function is constructed:
\begin{equation}
\label{eq:chassis_sign}
\constraint_{\text{chassis}}(\mathbf{x}) = -\bigl[\bar{C}_{\alpha f} \bar{C}_{\alpha r} (l_f + l_r)^2 + m v_x^2 (\bar{C}_{\alpha r} l_r - \bar{C}_{\alpha f} l_f)\bigr]
\end{equation}
This function defines a signed stability margin in the lateral-yaw phase plane, satisfying $\constraint_{\text{chassis}}(\mathbf{x}) < 0$ when the vehicle state lies inside the stable handling region, $\constraint_{\text{chassis}}(\mathbf{x}) > 0$ when it enters the unstable region, and $\constraint_{\text{chassis}}(\mathbf{x}) = 0$ on the stability boundary.

\subsection{Unified Motion Safety Constraints}

The environment constraint defined in Section~\ref{sec:geometric_constraint} and the chassis-level constraint defined above are combined to form a unified motion safety constraint.
For numerical consistency, the two signed safety functions are first normalized by their respective positive characteristic scaling factors:
\begin{equation}
\label{eq:normalize}
\bar{\constraint}_{\text{env}}(\mathbf{x})=\frac{\constraint_{\text{env}}(\mathbf{x})}{s_{\text{env}}},\qquad
\bar{\constraint}_{\text{chassis}}(\mathbf{x})=\frac{\constraint_{\text{chassis}}(\mathbf{x})}{s_{\text{chassis}}}
\end{equation}
where $\mathbf{x} \in \R^n$ denotes the vehicle state vector, and $s_{\text{env}}>0$ and $s_{\text{chassis}}>0$ are chosen to bring both constraints to a comparable order of magnitude.
Since the scaling factors are strictly positive, the sign and zero-level boundary of each constraint are preserved.
The unified motion safety constraint is then constructed as:
\begin{equation}
\label{eq:smooth_max}
\constraint(\mathbf{x}) =
\frac{1}{\kappa}\log\!\left(
e^{\kappa\bar{\constraint}_{\text{env}}(\mathbf{x})}
+ e^{\kappa\bar{\constraint}_{\text{chassis}}(\mathbf{x})}
\right)
\end{equation}
where $\kappa>0$ is the smoothing parameter.
This softmax-type construction provides a smooth approximation of the maximum operator while preserving the dominant safety constraint.
The unified motion safety set constrained by both environment and chassis dynamics is defined as:
\begin{equation}
\label{eq:unified_safety_set}
\safeSet = \bigl\{\mathbf{x} \in \R^n \mid \constraint(\mathbf{x}) \leq 0\bigr\}
\end{equation}

This region represents the vehicle's admissible state set under the combined effect of external collision-avoidance constraints and internal stability constraints, 
with its boundary jointly reflecting kinematic drivability and dynamic stability.
However, due to the non-convex and coupled nature of the constraints, it becomes challenging to perform safety analysis under extreme driving conditions and to provide reliable guidance for decision-making and control in autonomous driving systems.
Therefore, a generalizable reachability-based framework is introduced in the next section, which leverages Hamilton-Jacobi based reachability theory and offline reinforcement learning to approximate the motion safety set under realistic vehicle driving conditions.

\section{Modeling and Offline Reinforcement Learning-Based Approximation of Motion Safety Sets via Hamilton-Jacobi Reachability}
\label{sec:hj_reachability}

\subsection{Modeling of Motion Safety Sets Based on Hamilton-Jacobi Reachability}

To evaluate the vehicle's safety under aggressive maneuvers and potential disturbances, this section builds upon the unified constraint function to formulate a long-term motion safety framework.
Recalling the safety set defined in \eqref{eq:unified_safety_set},
this set represents the instantaneous safety status of the vehicle.
Under extreme driving conditions, long-term safety depends on the vehicle's future evolution.
Therefore, it is necessary to construct a motion safety model that incorporates both forward evolution and disturbances.

A vehicle dynamic system subject to disturbances can be expressed as:
\begin{equation}
\label{eq:disturbed_system}
\dot{\mathbf{x}}(t) = f\bigl(\mathbf{x}(t), \mathbf{u}(t), \mathbf{d}(t)\bigr), \quad \mathbf{u}(t) \in \mathcal{U}, \quad \mathbf{d}(t) \in \mathcal{D}
\end{equation}
where $\mathcal{U} \subset \R^m$ denotes the control input set, and $\mathcal{D} \subset \R^{m_d}$ denotes the disturbance set caused by model uncertainties or road excitation.
Given an initial state $\mathbf{x}_0 \in \mathcal{X}$ and a time horizon $[0,t]$, the system trajectory evolves as:
\begin{equation}
\label{eq:trajectory}
\mathbf{x}(\tau) = \zeta(\tau; 0, \mathbf{x}_0, \mathbf{u}(\cdot), \mathbf{d}(\cdot)), \quad \tau \in [0,t]
\end{equation}
where $\zeta(\cdot)$ denotes the state trajectory of the system under the joint influence of control and disturbance inputs.

\begin{definition}[Backward Reachable Tube]
\label{def:brt}
The backward reachable tube (BRT) of system \eqref{eq:disturbed_system} at time $t$ is defined as:
\begin{multline}
\label{eq:brt}
\mathcal{B}(t) = \bigl\{\mathbf{x}_0 \mid \exists\, \mathbf{d}(\cdot),\; \forall\, \mathbf{u}(\cdot),\; \exists\, \tau \in [0,t]:\\
\zeta(\tau; 0, \mathbf{x}_0, \mathbf{u}(\cdot), \mathbf{d}(\cdot)) \in \unsafeSet\bigr\}
\end{multline}
\end{definition}

If the system starts from an initial state $\mathbf{x}_0 \in \mathcal{B}(t)$, and there exists a disturbance signal such that no matter which control is applied, the trajectory eventually enters the unsafe set $\unsafeSet$, then $\mathbf{x}_0 \in \mathcal{B}(t)$.
This set characterizes states that are guaranteed to violate the safety constraints within the future time horizon.

\begin{definition}[Motion Safety Set]
\label{def:safety_set}
Based on the backward reachable tube, the finite-horizon motion safety set is defined as its complement:
\begin{multline}
\label{eq:safety_set}
\mathcal{S}(t) = \bigl(\mathcal{B}(t)\bigr)^c = \bigl\{\mathbf{x}_0 \mid \exists\, \mathbf{u}(\cdot),\; \forall\, \mathbf{d}(\cdot),\; \forall\, \tau \in [0,t]:\\
\zeta(\tau; 0, \mathbf{x}_0, \mathbf{u}(\cdot), \mathbf{d}(\cdot)) \in \safeSet\bigr\}
\end{multline}
\end{definition}

That is, starting from state $\mathbf{x}_0$, if there exists a control signal such that the system remains within $\safeSet$ over the entire interval $[0,t]$ in the presence of admissible disturbance signals, then $\mathbf{x}_0 \in \mathcal{S}(t)$.

Since the backward reachable tube $\mathcal{B}(t)$ expands over time, the motion safety set shrinks monotonically.
In practice, a sufficiently long-term horizon can approximate the infinite-horizon safety set:
\begin{equation}
\label{eq:infinite_safety}
\mathcal{S}_\infty = \lim_{t \to \infty} \mathcal{S}(t)
\end{equation}

To convert the motion safety set $\mathcal{S}(t)$ into a form that is computationally tractable, differentiable and safety-informative, we formulate a forward Hamilton-Jacobi (HJ) value function based on the constraint function $h(\mathbf{x})$.

\begin{definition}[Hamilton-Jacobi Value Function]
\label{def:hj_value}
The HJ value function is defined as the solution of the following Hamilton-Jacobi-Isaacs (HJI) PDE:
\begin{equation}
\label{eq:hj_value}
\begin{cases}
V(0, \mathbf{x}) = h(\mathbf{x}) \\
\displaystyle \frac{\partial V(t, \mathbf{x})}{\partial t} + \max_{\mathbf{d} \in \mathcal{D}} \min_{\mathbf{u} \in \mathcal{U}} \nabla_{\mathbf{x}} V(t, \mathbf{x})^\top f(\mathbf{x}, \mathbf{u}, \mathbf{d}) = 0, \quad t \in [0, T]
\end{cases}
\end{equation}
\end{definition}

Equation~\eqref{eq:hj_value} describes the dynamic propagation of constraint violations under admissible control and disturbance inputs.
Based on the HJ value function, the backward reachable tube and motion safety set are expressed as:
\begin{align}
\mathcal{B}(t) &= \bigl\{\mathbf{x} \in \mathcal{X} \mid V(t, \mathbf{x}) > 0\bigr\} \label{eq:brt_hj} \\
\mathcal{S}(t) &= \bigl\{\mathbf{x} \in \mathcal{X} \mid V(t, \mathbf{x}) \leq 0\bigr\} \label{eq:safety_hj}
\end{align}

This means that $V(t, \mathbf{x}) \leq 0$ implies the existence of a control strategy starting from state $\mathbf{x}$ that ensures $h(\mathbf{x}) < 0$ is satisfied at all future times.

In practical engineering applications, systems typically operate over a finite time horizon, and solving the infinite-horizon problem numerically is often intractable.
Therefore, a finite time horizon $T$ can be selected based on the characteristics of the specific driving scenario.
This time domain broadly covers key dynamic processes of the system, such as the complete time required for a vehicle to avoid an approaching hazard.
The engineering safety set is then defined as:
\begin{equation}
\label{eq:eng_safety}
\mathcal{S}_{\mathrm{eng}} = \mathcal{S}(T) = \bigl\{\mathbf{x} \in \mathcal{X} \mid V(T, \mathbf{x}) \leq 0\bigr\}
\end{equation}

Although the continuous-time HJ value function provides an analytic and differential representation of the motion safety set, its numerical solution suffers from the curse of dimensionality in high-dimensional state spaces.

To this end, we extend the problem to a discrete-time system and introduce the reachability value function $V_h^\pi$ and action-value function $Q_h^\pi$.

Assume sampling interval is $h$, the discrete time steps are $k = 0, 1, \ldots, N$, and the system is described as:
\begin{equation}
\label{eq:discrete_system}
\mathbf{x}_{k+1} = f(\mathbf{x}_k, \mathbf{u}_k, \mathbf{d}_k)
\end{equation}

In the discrete-time system, define the reachability value and action-value functions under policy $\pi$:
\begin{align}
V_h^\pi(\mathbf{x}) &= \max_{0 \leq s \leq N} h(\mathbf{x}_s) \;\big|\; \mathbf{x}_0 = \mathbf{x} \label{eq:disc_value} \\
Q_h^\pi(\mathbf{x}, \mathbf{u}) &= \max_{0 \leq s \leq N} h(\mathbf{x}_s) \;\big|\; \mathbf{x}_0 = \mathbf{x},\; \mathbf{u}_0 = \mathbf{u} \label{eq:disc_q}
\end{align}
Here, $V_h^\pi(\mathbf{x})$ defines the maximum constraint violation from state $\mathbf{x}$ under a given policy $\pi$, and $Q_h^\pi(\mathbf{x}, \mathbf{u})$ further quantifies safety performance under a specific maneuver $\mathbf{u}$.

Furthermore, the policy-optimal reachability value and action-value functions are:
\begin{align}
V_h^*(\mathbf{x}) &= \inf_{\pi} V_h^\pi(\mathbf{x}) \label{eq:opt_value} \\
Q_h^*(\mathbf{x}, \mathbf{u}) &= \inf_{\pi} Q_h^\pi(\mathbf{x}, \mathbf{u}) \label{eq:opt_q}
\end{align}

$V_h^*(\mathbf{x}) \leq 0$ implies the existence of a policy such that the system satisfies $h(\mathbf{x}_k) \leq 0$ throughout the finite horizon.

Furthermore, the finite-horizon motion safety set under continuous dynamics is approximately equivalent to the optimal safety set under discrete-time dynamics with engineering precision:
\begin{equation}
\label{eq:continuous_discrete_equiv}
\mathcal{S}(T) = \bigl\{\mathbf{x} \in \mathcal{X} \mid V(T, \mathbf{x}) \leq 0\bigr\} \approx \bigl\{\mathbf{x} \in \mathcal{X} \mid V_h^*(\mathbf{x}) \leq 0\bigr\} = \mathcal{S}_\pi
\end{equation}

The rigorous proof of consistency and convergence between the continuous and discrete safety sets is provided in~\cite{mitchell2005hj}.
This equivalence establishes that the discrete-time optimal reachability value function $V_h^*$ serves as a valid approximation of the continuous-time safety set, thereby enabling the use of data-driven methods for safety set computation.

\subsection{Construction of the Discounted Bellman Operator}

Although Section~\ref{sec:hj_reachability}-A has established that the motion safety set can be exactly characterized by the HJ reachability value function, grid-based numerical methods suffer from the curse of dimensionality.
To address this, we develop an offline RL framework that directly approximates $V_h^*$ and $Q_h^*$ from offline data via a reachability Bellman operator, without requiring online interaction or explicit policy construction.

The optimal reachability action-value function satisfies the following Bellman recursion:
\begin{equation}
\label{eq:bellman_q}
Q_h^*(\mathbf{x}, \mathbf{u}) = \sup_{\mathbf{d} \in \mathcal{D}} \max\biggl[\constraint(\mathbf{x}),\; \inf_{\mathbf{u}' \in \mathcal{U}} Q_h^*\bigl(f(\mathbf{x}, \mathbf{u}, \mathbf{d}), \mathbf{u}'\bigr)\biggr]
\end{equation}

where $\constraint(\mathbf{x})$ represents the instantaneous constraint violation and the second term characterizes the future violation under disturbance $\mathbf{d}$, mitigated by optimal action selection $\mathbf{u}'$.
The optimal reachability value function is:
\begin{equation}
\label{eq:bellman_v}
V_h^*(\mathbf{x}) = \inf_{\mathbf{u} \in \mathcal{U}} Q_h^*(\mathbf{x}, \mathbf{u})
\end{equation}

Equations \eqref{eq:bellman_q} and \eqref{eq:bellman_v} form the Bellman equations for the discrete-time reachability problem, where $\sup$ captures the \emph{worst-case} disturbance, $\inf$ reflects the controller's objective, and $\max$ accounts for current constraint violations.

However, the original Bellman equation is nonsmooth and lacks contraction properties.
Following~\cite{fisac2019bridging}, a discount factor $\gamma \in (0,1)$ is introduced to construct a discounted operator that ensures convergence while preserving structural consistency.

\begin{definition}[Discounted Bellman Operator]
\label{def:discounted_bellman}
For any bounded reachability action-value function $Q$, define discounted Bellman operator $\mathcal{P}^*$:
\begin{multline}
\label{eq:discounted_bellman}
(\mathcal{P}^* Q)(\mathbf{x}, \mathbf{u}) := (1-\gamma)\, h(\mathbf{x}) \\
+ \gamma \sup_{\mathbf{d} \in \mathcal{D}} \max\biggl[
h(\mathbf{x}),\;
\inf_{\mathbf{u}' \in \mathcal{U}}
Q\bigl(f(\mathbf{x}, \mathbf{u}, \mathbf{d}), \mathbf{u}'\bigr)
\biggr]
\end{multline}
\end{definition}

The corresponding policy-optimal discounted reachability action-value function is defined as the fixed-point solution:
\begin{equation}
\label{eq:q_fixed_point}
Q_{h,\gamma}^* = \mathcal{P}^* Q_{h,\gamma}^*
\end{equation}

The associated value function is:
\begin{equation}
\label{eq:v_gamma}
V_{h,\gamma}^*(\mathbf{x}) = \min_{\mathbf{u} \in \mathcal{U}} Q_{h,\gamma}^*(\mathbf{x}, \mathbf{u})
\end{equation}

As the discount factor $\gamma \to 1$, the discounted equation converges to the original continuous-time HJ reachability equation.

\begin{theorem}[Contraction Property of the Discounted Reachability Bellman Operator]
\label{thm:contraction}
On the space of bounded reachability action-value functions $\mathcal{F} = \{Q: \mathcal{X} \times \mathcal{U} \to \R \mid \|Q\|_\infty < \infty\}$, the discounted reachability Bellman operator $\mathcal{P}^*$ satisfies:
\begin{equation}
\label{eq:contraction}
\|\mathcal{P}^* Q_1 - \mathcal{P}^* Q_2\|_\infty \leq \gamma \|Q_1 - Q_2\|_\infty, \quad \forall\, Q_1, Q_2 \in \mathcal{F}
\end{equation}
\end{theorem}

\emph{Proof:} See Appendix~\ref{app:contraction}.

Since $\mathcal{P}^*$ is a $\gamma$-contraction mapping, the Banach fixed-point theorem guarantees the existence and uniqueness of $Q_{h,\gamma}^* = \mathcal{P}^* Q_{h,\gamma}^*$, and the value iteration converges geometrically.
These results establish the discounted reachability Bellman operator as a convergent framework for iterative value-function approximation.
The next subsection presents the practical offline RL algorithm that instantiates this framework to learn $Q_{h,\gamma}^*$ and $V_{h,\gamma}^*$ from extreme-driving data.

\subsection{Reachability Value Function Approximation Based on Offline Reinforcement Learning}

Following the theoretical analysis of the discounted reachability Bellman operator, practical approximation of reachability value functions in high-dimensional vehicle dynamics is considered.
Due to the prohibitive complexity of grid-based HJ solvers and the absence of online interaction under extreme driving conditions, the problem is formulated within an offline reinforcement learning framework that directly approximates $Q_{h,\gamma}^*(\mathbf{x},\mathbf{u})$ and $V_{h,\gamma}^*(\mathbf{x}) = \min_{\mathbf{u} \in \mathcal{U}} Q_{h,\gamma}^*(\mathbf{x},\mathbf{u})$.
Unlike conventional policy-optimization approaches, the proposed method performs pure value-function regression under the HJ reachability operator, using large-scale extreme driving data.

The training of the reachability action-value function $Q_{h,\gamma}$ can be interpreted as an offline regression toward the fixed point $Q_{h,\gamma}^*$ of the discounted reachability operator.
At each iteration, transition tuples $(\mathbf{x}, \mathbf{u}, h(\mathbf{x}), \mathbf{x}')$ are sampled from the offline dataset to construct the discounted reachability temporal-difference target:
\begin{equation}
\label{eq:td_target}
y_{Q,\gamma}(\mathbf{x}, \mathbf{u}) = (1-\gamma)\, h(\mathbf{x}) + \gamma \max\bigl\{h(\mathbf{x}),\; V_{h,\gamma}(\mathbf{x}')\bigr\}
\end{equation}

The reachability action-value function is updated by minimizing the temporal-difference error, where $\ell(\cdot)$ denotes the squared loss function:
\begin{equation}
\label{eq:q_loss}
\mathcal{L}_{Q_{h,\gamma}} = \mathbb{E}_{(\mathbf{x}, \mathbf{u}, \mathbf{x}') \sim \mathcal{D}}\bigl[\ell\bigl(y_{Q,\gamma}(\mathbf{x}, \mathbf{u}) - Q_{h,\gamma}(\mathbf{x}, \mathbf{u})\bigr)\bigr]
\end{equation}

After updating $Q_{h,\gamma}$, it is projected to the state value function $V_{h,\gamma}$.
Since directly minimizing over all actions suffers from out-of-distribution optimistic bias, reversed expectile regression (RER) is employed as a left-tail-sensitive regression to suppress unrealistically low action-value estimates.
The update loss is:
\begin{equation}
\label{eq:v_loss}
\mathcal{L}_{V_{h,\gamma}} = \mathbb{E}_{(\mathbf{x},\mathbf{u})\sim\mathcal{D}}\bigl[L_{\text{rev}}^{\tau}\bigl(Q_{h,\gamma}(\mathbf{x},\mathbf{u}) - V_{h,\gamma}(\mathbf{x})\bigr)\bigr]
\end{equation}
where the reversed expectile regression loss is given by:
\begin{equation}
\label{eq:rer_loss}
L_{\text{rev}}^{\tau}(\xi) = |\tau_{\text{rev}} - \mathbb{I}(\xi > 0)| \,\xi^2, \quad \tau_{\text{rev}} \in (0.5,1)
\end{equation}
Here, $\xi = Q_{h,\gamma}(\mathbf{x},\mathbf{u}) - V_{h,\gamma}(\mathbf{x})$ denotes the value residual.

Since smaller reachability values indicate safer states, $\xi < 0$ is equivalent to $Q_{h,\gamma} < V_{h,\gamma}$ and corresponds to a safer action value; this case is weighted by $\tau_{\text{rev}} > 0.5$, pulling $V_{h,\gamma}$ downward toward safer actions.
Conversely, $\xi > 0$ is equivalent to $Q_{h,\gamma} > V_{h,\gamma}$ and corresponds to a riskier action value; this case is weighted by $1-\tau_{\text{rev}} < 0.5$, weakening its upward influence on $V_{h,\gamma}$.
Therefore, as $\tau_{\text{rev}} \to 1$, $V_{h,\gamma}$ increasingly tracks the minimum action value across the data distribution, approaching the true reachability value function $\min_{\mathbf{u}} Q_{h,\gamma}$.

An alternating training algorithm is constructed to jointly optimize $Q_{h,\gamma}$ and $V_{h,\gamma}$, as summarized in Table~\ref{tab:algorithm}.
At each iteration, the state value function is first fixed to construct the TD target for updating $Q_{h,\gamma}$; then with $Q_{h,\gamma}$ fixed, $V_{h,\gamma}$ is updated by minimizing the RER loss.
This procedure can be interpreted as a conservative variant of Fitted Q Iteration (FQI) under the reachability Bellman operator.
For numerical stability, independent MLP architectures are adopted for $Q_{h,\gamma}$ and $V_{h,\gamma}$ with normalized inputs.
The Adam optimizer and gradient clipping are applied to mitigate oscillations.
A boundary-aware sample reweighting strategy assigns higher weights to samples with small $|V_{h,\gamma}(\mathbf{x})|$, focusing the regression on states near the safety boundary where $V_{h,\gamma}(\mathbf{x}) \approx 0$.

\begin{table}[t]
\caption{Offline Reachability Value Function Learning}
\label{tab:algorithm}
\centering
\begin{tabular}{p{0.95\columnwidth}}
\toprule
\textbf{Input:} Offline dataset $\mathcal{D}$, discount factor $\gamma$, RER parameter $\tau_{\text{rev}}$, learning rate $\eta$ \\
\midrule
Randomly initialize the parameters of the action-value network $Q_{h,\gamma}$ and the state-value network $V_{h,\gamma}$ \\
\textbf{For} $i = 1$ to $N_{\text{iter}}$ \\
\quad Randomly sample a mini-batch from the offline dataset $\mathcal{D}$ \\
\quad Apply boundary-aware sample reweighting: $w(\mathbf{x}) = 1 / (|V_{h,\gamma}(\mathbf{x})| + \epsilon_w)$ \\
\quad Construct the discounted reachability target: \\
\quad $y_{Q,\gamma} = (1-\gamma)\, h(\mathbf{x}) + \gamma \max\bigl\{h(\mathbf{x}),\; V_{h,\gamma}(\mathbf{x}')\bigr\}$ \\
\quad Compute the temporal-difference loss: $\mathcal{L}_{Q_{h,\gamma}}$ \\
\quad Update $Q_{h,\gamma}$ using learning rate $\eta$ \\
\quad Perform RER-based projection to update $V_{h,\gamma}$: \\
\quad\quad Compute state-action residual: $\xi = Q_{h,\gamma}(\mathbf{x}, \mathbf{u}) - V_{h,\gamma}(\mathbf{x})$ \\
\quad\quad Define asymmetric RER weight: $w_{\tau} = |\tau_{\text{rev}} - \mathbb{I}(\xi > 0)|$ \\
\quad\quad Compute state-value regression loss: $\mathcal{L}_{V_{h,\gamma}}$ \\
\quad\quad Update $V_{h,\gamma}$ using learning rate $\eta$ \\
\textbf{End for} \\
\textbf{Output:} Trained $Q_{h,\gamma}$ and $V_{h,\gamma}$; construct motion safety set \\
\bottomrule
\end{tabular}
\end{table}

\section{Safe Reinforcement Learning for Collision Avoidance with Motion Safety Set Constraints}
\label{sec:safe_rl_control}

\subsection{Safe Reinforcement Learning Problem Formulation}
\label{sec:safe_rl_formulation}

Standard reinforcement learning formulates the decision-making problem as a Markov decision process (MDP) defined by the tuple $(\mathcal{O}, \mathcal{A}, P, r, \gamma_r)$, where $\mathcal{O}$ is the observation space, $\mathcal{A}$ is the action space, $P$ is the transition dynamics, $r$ is the reward function, and $\gamma_r \in (0,1)$ is the discount factor.
The objective is to find an optimal policy $\pi^*$ that maximizes the expected cumulative reward $J_r(\pi) = \mathbb{E}_{\pi}\bigl[\sum_{t=0}^{\infty} \gamma_r^t r_t\bigr]$.
However, in safety-critical tasks such as autonomous driving collision avoidance, addressing safety solely through reward shaping is insufficient, as it lacks formal safety guarantees to prevent catastrophic outcomes.

To incorporate safety requirements, the emergency collision avoidance problem is reformulated as a constrained Markov decision process (CMDP) $(\mathcal{O}, \mathcal{A}, P, r, c, \gamma_r, \gamma_c)$, where $c: \mathcal{O} \to [0,1]$ is a bounded safety cost function, $\gamma_c \in (0,1)$ is the cost discount factor, and $d_c > 0$ is the acceptable cost threshold.
The safe RL objective is then:
\begin{equation}
\label{eq:cmdp}
\begin{aligned}
\pi^* = \arg\max_{\pi} \; &J_r(\pi) \\
\text{s.t.} \quad &J_c(\pi) = \mathbb{E}_{\pi}\biggl[\sum_{t=0}^{\infty} \gamma_c^t c(\mathbf{o}_t)\biggr] \leq d_c
\end{aligned}
\end{equation}
where $J_c(\pi)$ denotes the expected cumulative safety cost, and the constraint requires that the long-term cost remains below the threshold $d_c$.

In this work, the learned motion safety set from Section~\ref{sec:hj_reachability} is leveraged to construct a unified, forward-looking continuous safety cost function $c$, as detailed in Section~\ref{sec:safe_rl_env}.

\subsection{Safe RL Environment Design}
\label{sec:safe_rl_env}

To instantiate the CMDP formulation above, a safe reinforcement learning environment is designed for emergency collision-avoidance scenarios, building upon the vehicle dynamics model and motion safety set established in Sections~\ref{sec:modeling}--\ref{sec:hj_reachability}.

\textbf{Scenario Description.}
The training scenario considers forward obstacle avoidance on a straight road under low-adhesion conditions.
The ego vehicle travels at a reference longitudinal speed $v_{\text{ref}}$ along a road segment with lateral boundaries $Y \in [Y_{\min}, Y_{\max}]$.
An oriented rectangular obstacle, parameterized by position $\mathbf{p}_{\text{obs}}$, length $l_{\text{obs}}$, width $w_{\text{obs}}$, and heading $\psi_{\text{obs}}$, is placed along the road.
The ego vehicle must execute a collision-avoidance maneuver while maintaining lateral stability and staying within road boundaries.
Fig.~\ref{fig:scenario} illustrates the scenario configuration.
The collision detection between the ego vehicle and obstacles follows the SAT-based signed safety margin $h_{\mathrm{env}}(\mathbf{x})$ defined in \eqref{eq:henv}, and the road boundary constraint is evaluated based on the ego vehicle's oriented bounding box extent.

\begin{figure}[t]
\centering
\includegraphics[width=0.8\columnwidth]{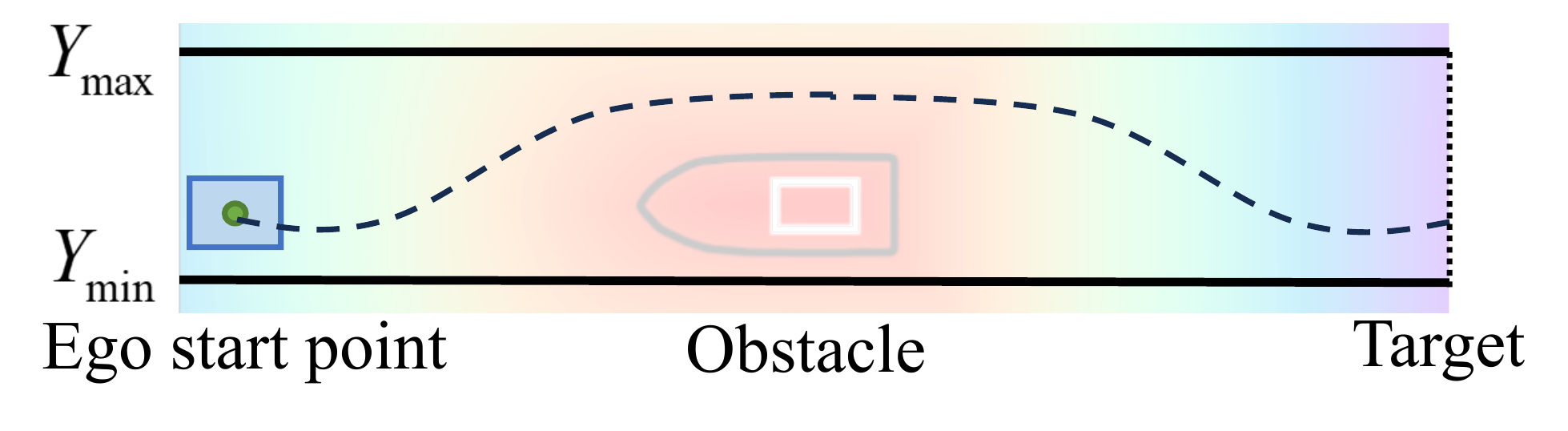}
\caption{Illustration of the forward obstacle avoidance scenario.}
\label{fig:scenario}
\end{figure}

\textbf{Observation and Action Spaces.}
The observation space is defined as:
\begin{equation}
\label{eq:rl_state}
\mathbf{o}_t = \bigl[v_x,\; v_y,\; r,\; X,\; Y,\; \psi,\; \Delta x_{\text{obs}},\; \Delta y_{\text{obs}},\; \Delta \psi_{\text{obs}}\bigr]^\top \in \mathcal{O}
\end{equation}
where $v_x$ and $v_y$ are the longitudinal and lateral velocities, $r$ is the yaw rate, $X$ and $Y$ denote the global position of the vehicle center of mass, $\psi$ is the yaw angle, and $\Delta x_{\text{obs}}$, $\Delta y_{\text{obs}}$, $\Delta \psi_{\text{obs}}$ denote the relative position and heading of the obstacle with respect to the ego vehicle body frame.

The action space consists of a normalized steering command:
\begin{equation}
\label{eq:rl_action}
a_t = \delta_{\text{norm}} \in [-1,\; 1]
\end{equation}
which maps to the physical front steering angle via $\delta_f = \delta_{\text{norm}} \cdot \delta_{\max}$.
The actual steering angle is filtered through a first-order inertial model to account for steering actuator dynamics and ensure physically feasible commands.
The longitudinal speed is regulated by a low-level speed controller at the reference velocity $v_{\text{ref}}$, so that the policy focuses on learning lateral collision-avoidance maneuvers.

\textbf{Reward Design.}
The reward function consists of a dense step reward and terminal bonuses/penalties.
The step reward is:
\begin{equation}
\label{eq:rl_reward}
r_t^{\text{step}} = w_1 - w_2 Y^2 - w_3 \psi^2 - w_4 (\Delta\delta_f)^2
\end{equation}
where $w_1$, $w_2$, $w_3$, $w_4$ are weighting coefficients.
The first term provides a constant survival bonus, while the remaining terms penalize lateral deviation from the road centerline, heading errors, and steering jerk, respectively.
Upon episode termination, the terminal reward is:
\begin{equation}
\label{eq:rl_terminal}
r^{\text{term}} = \begin{cases}
\bigl[R_{\text{target}} - w_y Y_T^2 - w_\psi \psi_T^2\bigr]_+ & \text{target reached} \\
-p_{\text{col}} & \text{collision} \\
-p_{\text{bnd}} & \text{boundary violation} \\
-p_{\text{yaw}} & \text{excessive heading}
\end{cases}
\end{equation}
where $R_{\text{target}}$ is the base target bonus, $Y_T$ and $\psi_T$ are the terminal lateral offset and heading, and $p_{\text{col}}$, $p_{\text{bnd}}$, $p_{\text{yaw}}$ are fixed penalty coefficients for collision, boundary violation, and excessive heading deviation respectively.

\textbf{Safety Cost Function.}
Rather than relying solely on instantaneous or locally defined safety indicators such as geometric collision flags or stability measures, the learned motion safety set is incorporated into the CMDP as the safety cost.
The reachability value function $V_{h,\gamma}$ characterizes the control-optimized upper bound of future constraint violation under disturbances, providing a forward-looking criterion for evaluating whether the current state may evolve into an unsafe region.
Therefore, the safety cost is designed as a normalized violation cost:
\begin{equation}
\label{eq:safety_constraint}
c(\mathbf{o}_t)
= \operatorname{clip}\left(
\frac{V_{h,\gamma}(\mathbf{o}_t)}{\epsilon_V},
\,0,\,1
\right), \quad \epsilon_V > 0 .
\end{equation}
here, $\epsilon_V$ controls the value range over which the violation cost increases from zero to one.
The resulting cost has the following interpretation:
\[
c(\mathbf{o}_t)=
\begin{cases}
0, & V_{h,\gamma}(\mathbf{o}_t) \leq 0,\\
\dfrac{V_{h,\gamma}(\mathbf{o}_t)}{\epsilon_V},
& 0 < V_{h,\gamma}(\mathbf{o}_t) < \epsilon_V,\\
1, & V_{h,\gamma}(\mathbf{o}_t) \geq \epsilon_V .
\end{cases}
\]
The cost therefore remains zero for all safe states, increases continuously only after the state leaves the learned motion safety set, and saturates at one for large positive reachability values.

The motion-safety-set based formulation offers important advantages over conventional hand-crafted constraints.
First, $V_{h,\gamma}$ provides a unified safety representation that incorporates both external collision risk and internal chassis instability into a single criterion, avoiding the use of multiple separately designed constraints during policy optimization.
Second, by encoding future safety evolution under disturbances and admissible control responses, $V_{h,\gamma}$ defines a forward-looking safety boundary for constrained policy optimization.
The clipped continuous cost also preserves violation severity information that would be discarded by a binary indicator, allowing the cost critic to distinguish slight safety-set violations from severe unsafe states.

\subsection{Safe Policy Optimization}
\label{sec:policy_optimization}

With the safety cost function defined in \eqref{eq:safety_constraint}, the constrained policy optimization problem in \eqref{eq:cmdp} is solved using a PID-Lagrangian Soft Actor-Critic framework.
The safety constraint is incorporated into the maximum-entropy objective as
\begin{equation}
\label{eq:lagrangian}
\max_{\theta} \min_{\lambda \geq 0} \bigl[ J_{\alpha}(\pi_\theta) - \lambda (J_c(\pi_\theta) - d_c) \bigr]
\end{equation}
where $J_{\alpha}$ is the entropy-regularized return, $J_c$ is the expected discounted safety cost, and $d_c$ is the prescribed safety threshold.

The overall architecture is shown in Fig.~\ref{fig:sac_lag}.
\begin{figure}[t]
\centering
\includegraphics[width=\columnwidth]{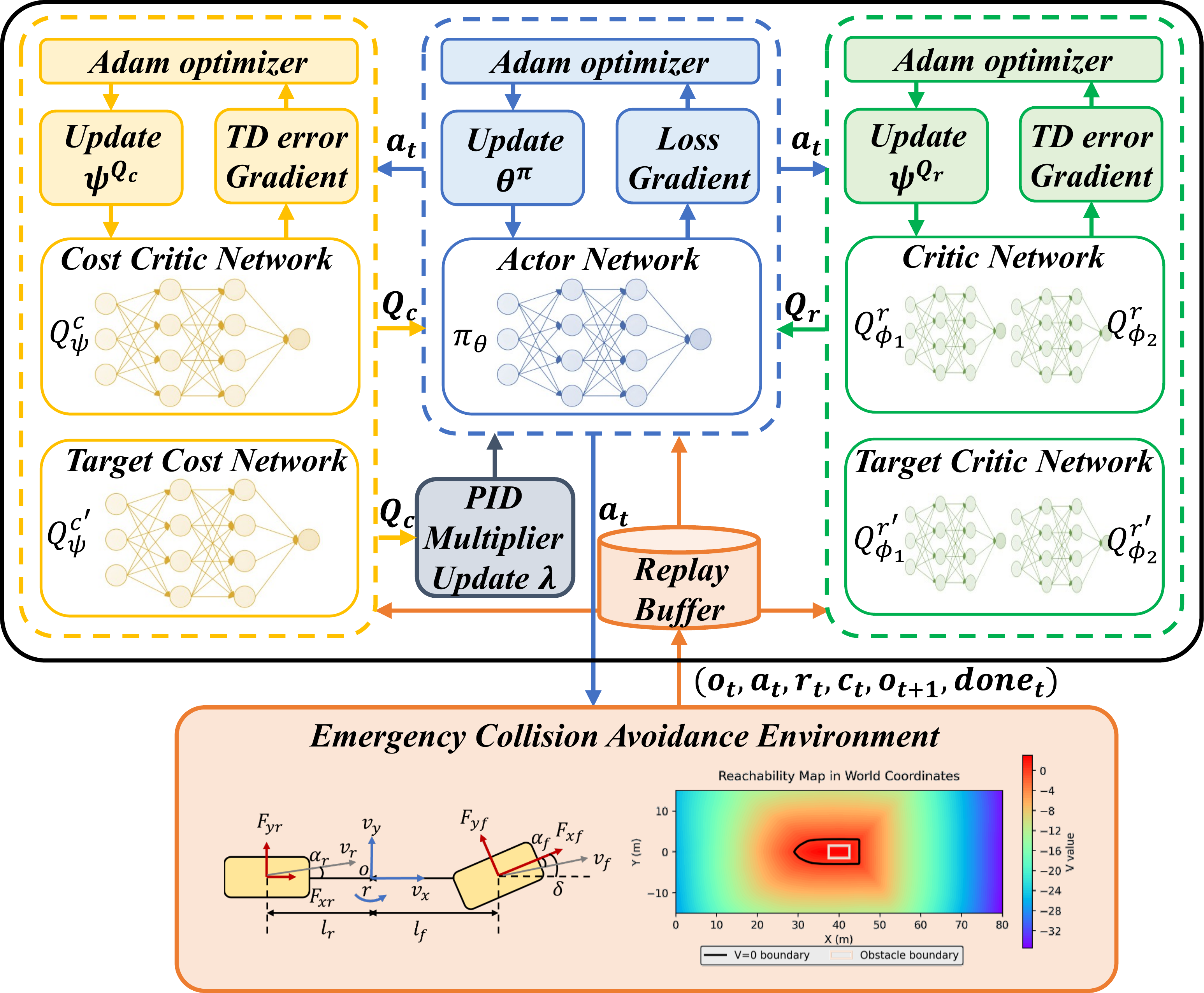}
\caption{Network architecture of the PID-Lagrangian SAC algorithm. The actor $\pi_\theta$ outputs a Gaussian action distribution with tanh squashing. Clipped double reward critics $Q^r_{\phi_1}, Q^r_{\phi_2}$ and a single cost critic $Q^c_{\psi}$ are each paired with target networks updated via soft averaging. A scalar Lagrange multiplier $\lambda$ is adaptively adjusted via a PID controller to enforce the safety constraint.}
\label{fig:sac_lag}
\end{figure}

It consists of a stochastic actor $\pi_\theta$, clipped double reward critics $Q^r_{\phi_1}, Q^r_{\phi_2}$, a cost critic $Q^c_\psi$, and their target networks.
The reward critics and entropy temperature are updated following standard SAC, while the cost critic estimates the discounted cumulative motion-safety cost:
\begin{equation}
\label{eq:qc_def}
Q^c(\mathbf{o}_t, a_t) = \mathbb{E}_{\pi} \biggl[\sum_{k=t}^{\infty} \gamma_c^{k-t} c_k\biggr]
\end{equation}
where $c_k = c(\mathbf{o}_k)$ is defined in \eqref{eq:safety_constraint}.
Its Bellman target is
\begin{equation}
\label{eq:qc_target}
y_t^c = c_t + \gamma_c Q^c_{\bar{\psi}}(\mathbf{o}_{t+1}, a_{t+1}), \quad a_{t+1} \sim \pi_\theta(\cdot|\mathbf{o}_{t+1})
\end{equation}
and $Q^c_\psi$ is fitted by minimizing the corresponding squared Bellman error.

The actor is optimized by minimizing the Lagrangian SAC objective:
\begin{equation}
\label{eq:policy_loss}
\mathcal{L}_{\pi}(\theta) = \mathbb{E}\bigl[\alpha \log \pi_\theta(a_t|\mathbf{o}_t) - \min_{i=1,2} Q^r_{\phi_i}(\mathbf{o}_t, a_t) + \lambda Q^c_\psi(\mathbf{o}_t, a_t)\bigr]
\end{equation}

Thus, the policy is encouraged to maximize task return while penalizing actions with large predicted safety costs.

Instead of the naive gradient-based multiplier update, a PID-Lagrangian controller~\cite{stooke2020responsive} is used to regulate $\lambda$ according to the safety constraint violation.
Let
\begin{equation}
\label{eq:cost_error}
e_t = \hat{J}_c^{(t)} - d_c
\end{equation}
where $\hat{J}_c^{(t)}$ is the exponential moving average of the minibatch cost estimates.
The multiplier is updated as
\begin{equation}
\label{eq:pid_lambda}
\lambda \leftarrow \bigl[K_p \tilde{e}_t + I_t + K_d D_t\bigr]_+
\end{equation}
where $[\cdot]_+$ denotes the projection onto $\lambda \geq 0$, with
\begin{align}
\tilde{e}_t &= \alpha_p \tilde{e}_{t-1} + (1-\alpha_p) e_t \label{eq:pid_p} \\
I_t &= \max(0, I_{t-1} + K_i e_t) \label{eq:pid_i} \\
D_t &= \max(0, \hat{J}_c^{(t)} - \hat{J}_c^{(t-\ell)}). \label{eq:pid_d}
\end{align}
here, $\tilde{e}_t$ is the exponential moving average of the error with smoothing coefficient $\alpha_p$, and $\ell$ is the delay horizon for the derivative term.
The proportional, integral, and derivative terms respond to the current violation magnitude, accumulated violation, and cost trend, respectively, which improves transient response and reduces overshoot in safety-constraint regulation.

\section{Validation and Results}
\label{sec:validation}

The proposed framework is validated on the forward obstacle avoidance scenario described in Section~\ref{sec:safe_rl_env}.
The validation procedure consists of three stages: (i) real-vehicle extreme driving data are collected and used to approximate the motion safety set through offline reinforcement learning; (ii) the learned motion safety set is incorporated as the safety cost function, and the PID-Lagrangian SAC policy is trained in simulation with comparative and ablation studies; (iii) the trained policy is deployed on a production-level passenger vehicle to evaluate its collision-avoidance performance in real-vehicle experiments.

\subsection{Motion Safety Set Approximation and Visualization}
\label{sec:result_safety_set}

To support offline reinforcement learning--based motion safety set approximation for the forward obstacle avoidance scenario, a high-quality offline dataset was constructed from real-vehicle extreme driving experiments.
Each data sample is a tuple $(\mathbf{x}, \mathbf{u}, \mathbf{x}', h)$, where $\mathbf{x}$ and $\mathbf{x}'$ denote the current and next reachability states $\mathbf{x}_{\text{reach}} = [v_x,\; v_y,\; r,\; \Delta x_{\text{obs}},\; \Delta y_{\text{obs}},\; \Delta \psi_{\text{obs}}]^\top \in \mathbb{R}^6$ selected from the observation vector in \eqref{eq:rl_state}, $\mathbf{u}$ consists of the front steering angle, driving or braking torque, and $h(\mathbf{x})$ is the unified signed safety constraint value defined in \eqref{eq:smooth_max}.

\begin{figure}[t]
\centering
\includegraphics[width=0.9\columnwidth]{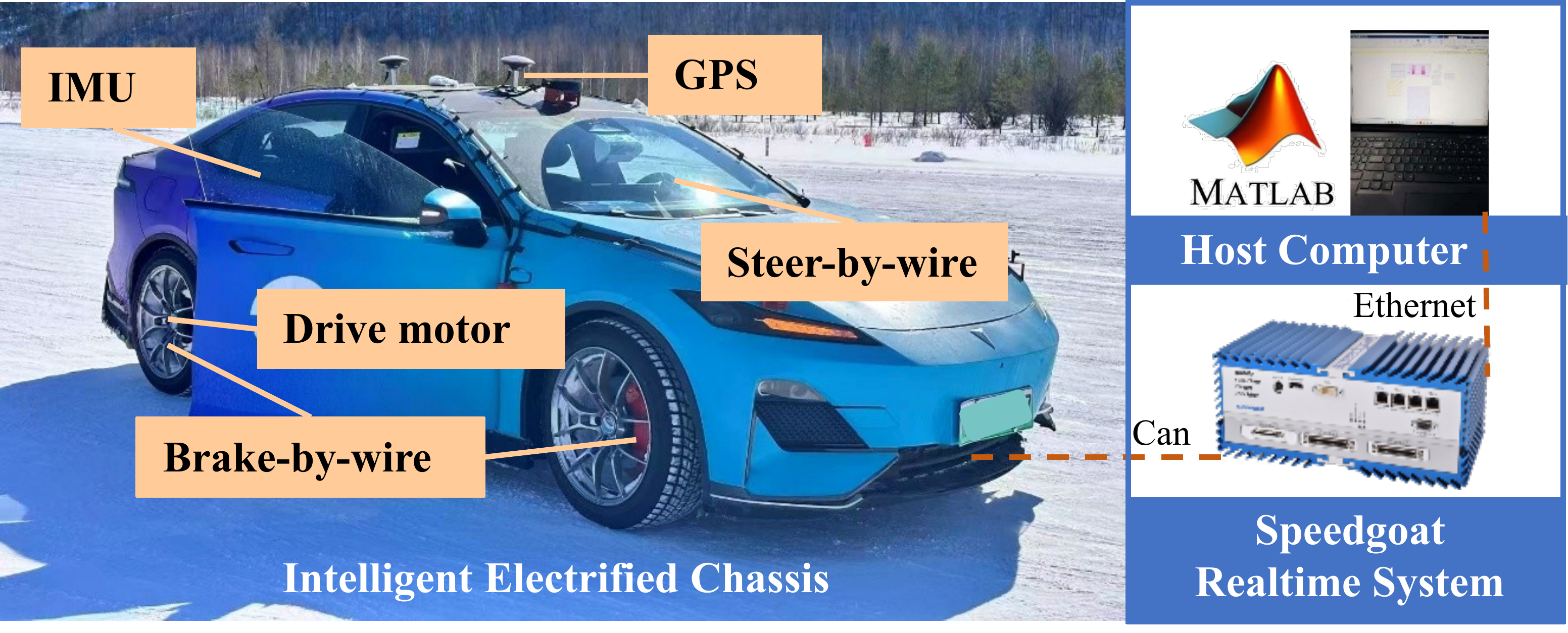}
\caption{The Changan Deepal SL03 production-level passenger vehicle used as the experimental platform.}
\label{fig:test_vehicle}
\end{figure}

A Changan Deepal SL03 passenger vehicle (Fig.~\ref{fig:test_vehicle}) equipped with onboard IMU and GPS sensors was employed as the experimental data acquisition platform, whose key parameters are listed in Table~\ref{tab:vehicle_params}.

\begin{table}[t]
\centering
\caption{Key vehicle parameters of the experimental platform (Changan Deepal SL03).}
\label{tab:vehicle_params}
\begin{tabular}{lcc}
\hline
\textbf{Parameter} & \textbf{Symbol} & \textbf{Value} \\
\hline
Total mass & $m$ & 2178 kg \\
Yaw moment of inertia & $I_z$ & 3216 kg$\cdot$m$^2$ \\
CG-to-front-axle distance & $l_f$ & 1.526 m \\
CG-to-rear-axle distance & $l_r$ & 1.374 m \\
CG height & $h_{\text{cg}}$ & 0.175 m \\
Tire model & --- & PS4 245/45 R19 \\
\hline
\end{tabular}
\end{table}
Data collection focuses on the forward obstacle avoidance scenario, where a stationary rectangular obstacle is placed ahead of the ego vehicle on a straight road.
Five professional test drivers conducted more than 300 collision-avoidance experiments under varying initial velocities, initial relative distances, and control constraint conditions.
During these experiments, the system recorded vehicle state variables and control signals, yielding a diverse set of extreme collision-avoidance maneuvering data.
After time synchronization and filtering, the unified safety constraint $h(\mathbf{x})$ was computed for each sample.
The states and actions were then normalized, and the data were further augmented through mirror transformation and random translation, resulting in approximately 4.4 million tuples $(\mathbf{x}, \mathbf{u}, \mathbf{x}', h)$, with approximately 65\% non-collision and 35\% collision trajectories.

\begin{figure}[t]
\centering
\includegraphics[width=\columnwidth]{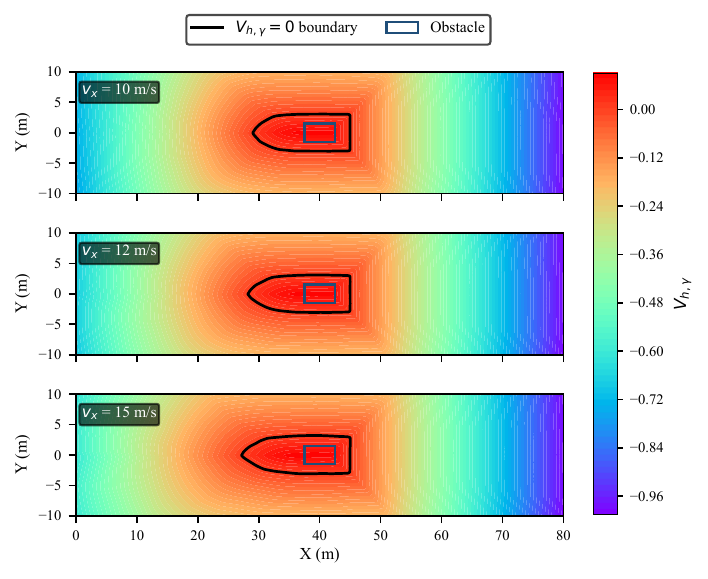}
\caption{Comparison of motion safety set boundaries under different longitudinal velocities ($\Delta\psi_{\text{obs}} = 0^{\circ}$). As $v_x$ increases, the backward reachable unsafe set expands mainly ahead of the obstacle due to increased forward momentum and longer braking distance, and equivalently the motion safety set shrinks.}
\label{fig:reachability_vx}
\end{figure}

\begin{figure}[t]
\centering
\includegraphics[width=\columnwidth]{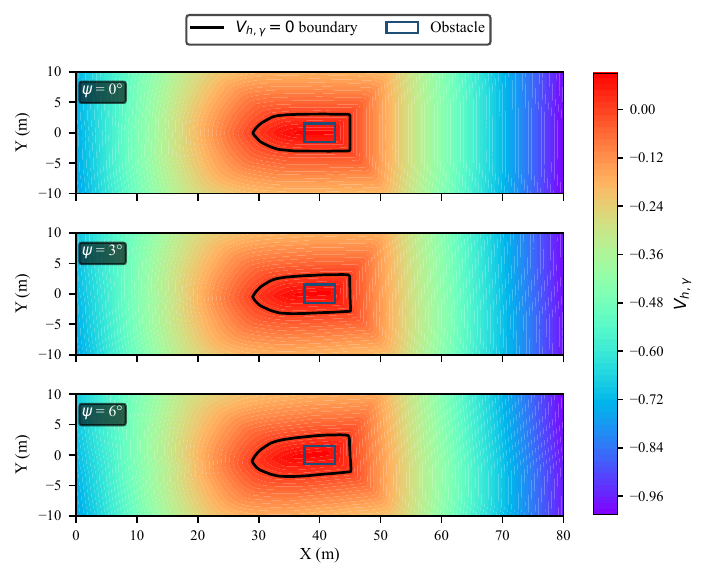}
\caption{Comparison of motion safety set boundaries under different relative yaw angles $\Delta\psi_{\text{obs}}$ ($v_x = 12$ m/s). Nonzero $\Delta\psi_{\text{obs}}$ introduces asymmetric coupling between lateral and yaw dynamics, causing the backward reachable unsafe set to deform and expand.}
\label{fig:reachability_yaw}
\end{figure}

The offline reachability value function learning algorithm described in Section~\ref{sec:hj_reachability} was applied to the collected dataset to approximate the discounted reachability value function $V_{h,\gamma}$.
The network architecture adopts independent MLPs for $Q_{h,\gamma}$ and $V_{h,\gamma}$, each with three hidden layers of 256 units and ReLU activations.
The discount factor is set to $\gamma = 0.99$, the RER parameter to $\tau_{\text{rev}} = 0.8$, and the learning rate to $\eta = 3 \times 10^{-4}$ with Adam optimizer and gradient clipping at $1.0$.
Training was conducted for 8 million iterations with a mini-batch size of 4096.

Both the ego vehicle and the obstacle are modeled as oriented rectangles with length $L = 5$ m and width $W = 3$ m for collision detection.
Since the reachability state space is six-dimensional, direct visualization of the motion safety set $\{x \mid V_{h,\gamma}(x) \leq 0\}$ is intractable.
Instead, two-dimensional slices are extracted by projecting onto the relative position plane $(\Delta x_{\text{obs}}, \Delta y_{\text{obs}})$, revealing how the safety boundary responds to key dynamic parameters.

Figs.~\ref{fig:reachability_vx} and~\ref{fig:reachability_yaw} present the motion safety set under different longitudinal velocities $v_x$ and relative yaw angles $\Delta\psi_{\text{obs}}$, respectively, at a road adhesion coefficient of $\mu = 0.3$, with $v_y = 0$ and $r = 0$ fixed in all panels.
In each panel, the outer black contour denotes the learned motion safety set boundary $V_{h,\gamma} = 0$, and the inner dark blue contour represents the physical obstacle boundary.
The color map represents the reachability value $V_{h,\gamma}$: darker colors indicate safer states with larger negative values, while regions approaching red correspond to higher reachability values, meaning smaller collision avoidance margins or more severe collision intrusion.

\textbf{Effect of Longitudinal Velocity.}
Fig.~\ref{fig:reachability_vx} compares the motion safety set at $v_x = 10$, $12$, and $15$ m/s with $\Delta\psi_{\text{obs}} = 0^{\circ}$ fixed.
The black contour represents the boundary of the backward reachable unsafe region induced by the obstacle and vehicle dynamics.
At low speed ($v_x = 10$ m/s), the unsafe reachable region remains relatively close to the physical obstacle boundary, indicating sufficient maneuverability for collision avoidance.
As $v_x$ increases to $12$ and $15$ m/s, this unsafe reachable region expands mainly in the longitudinal direction, and equivalently the admissible motion safety set, defined as its complement, shrinks.

\textbf{Effect of Relative Yaw Angle.}
Fig.~\ref{fig:reachability_yaw} shows the safety boundary under different $\Delta\psi_{\text{obs}}$ at the same longitudinal velocity.
Since the obstacle heading is fixed at zero, varying $\Delta\psi_{\text{obs}}$ is equivalent to varying the ego vehicle yaw angle.
At $\Delta\psi_{\text{obs}} = 0^{\circ}$, the set is approximately symmetric about the longitudinal axis.
As $|\Delta\psi_{\text{obs}}|$ increases, the boundary contracts asymmetrically on the side toward the heading direction, due to the coupling between lateral displacement and yaw dynamics that biases the effective steering authority.

Across all conditions, the motion safety set boundary consistently encloses a smaller region than the geometrically collision-free space.
Increasing longitudinal velocity leads to longitudinal contraction of the safe region due to greater braking distance, while nonzero relative yaw angle induces rotational and asymmetric deformation of the safety boundary.
These results confirm that the learned motion safety set captures forward-looking risk under disturbances that purely geometric proximity measures cannot reveal.

\subsection{Safe RL Training and Comparative Analysis}
\label{sec:result_safe_rl}

To evaluate the effectiveness of the proposed PID-Lagrangian SAC with motion safety set (MSS) guidance, three configurations are compared:

\begin{enumerate}
\item \textbf{PID-Lagrangian SAC with MSS (PID-Lag-SAC w/ MSS)}: The proposed method, where the learned motion safety set is used through the continuous safety cost $c$ and the Lagrange multiplier is regulated by a PID controller.
\item \textbf{PID-Lagrangian SAC without MSS (PID-Lag-SAC w/o MSS)}: The PID-Lagrangian structure is retained, but the reachability value function in the continuous safety cost is replaced by the instantaneous normalized unified safety constraint $\constraint(\mathbf{x}_t)$, thereby removing the forward-looking reachability criterion.
\item \textbf{Lagrangian SAC with MSS (Lag-SAC w/ MSS)}: The learned motion safety set is retained through $c$, but the PID controller is replaced by a standard gradient-based Lagrangian multiplier update $\lambda \leftarrow [\lambda + \eta_\lambda (J_c - d_c)]_+$.
\end{enumerate}

\textbf{Training Configuration.}
All three methods use the same SAC backbone and training hyperparameters.
The actor, reward critics, and cost critic are implemented as MLPs with two hidden layers of 256 units and ReLU activations.
The discount factors are set to $\gamma = \gamma_c = 0.99$, the learning rate is $3 \times 10^{-4}$ with Adam optimizer, the safety cost threshold is $d_c = 1.0$, and each method is trained for 1000 epochs with a batch size of 256.
The simulation is conducted in the forward obstacle-avoidance environment described in Section~\ref{sec:safe_rl_env} under a low-adhesion condition of $\mu = 0.3$.
Following the geometric modeling in Section~\ref{sec:result_safety_set}, both the ego vehicle and the obstacle are represented as oriented rectangles with length $L = 5$ m and width $W = 3$ m.
The stationary obstacle is placed at $(40,0)$ m with zero heading, and the post-avoidance target is set to $X_{\text{target}} = 80~\mathrm{m}$, requiring the policy to bypass the obstacle and return to the original lane.

\begin{figure}[t]
\centering
\includegraphics[width=\columnwidth]{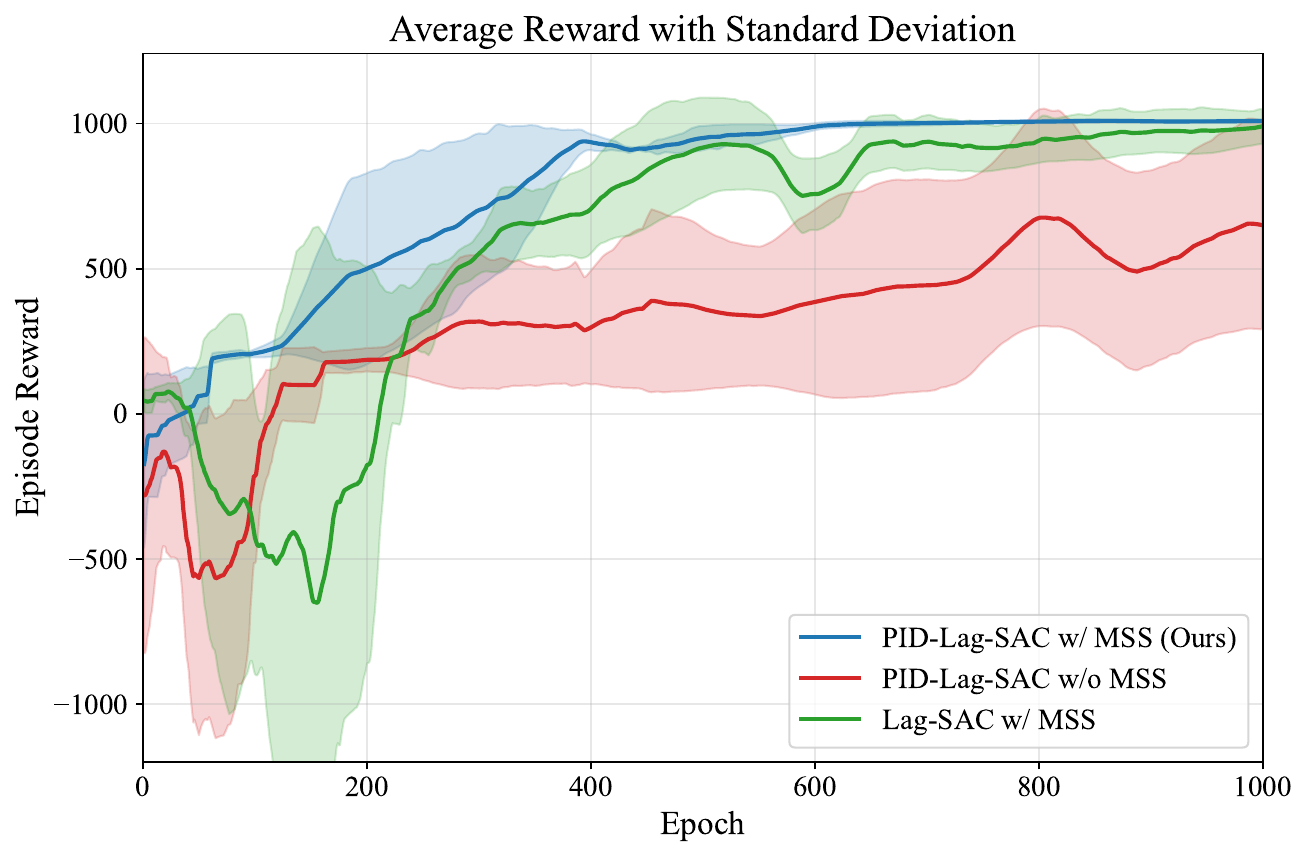}
\caption{Training reward curves of the three compared methods. }
\label{fig:training_reward}
\end{figure}

\textbf{Training Performance.}
Fig.~\ref{fig:training_reward} presents the average episode reward of the three compared methods during training.
The proposed PID-Lag-SAC w/ MSS exhibits superior performance in terms of convergence speed, training stability, and final reward magnitude.
After a short initial exploration phase, its reward increases rapidly and reaches a high-performance region within approximately 300--400 epochs, and then remains close to the highest reward level with relatively small deviation.
In comparison, Lag-SAC w/ MSS eventually approaches a comparable reward level but suffers from a pronounced early-stage performance drop and larger variance, suggesting that the standard gradient-based multiplier update leads to unstable constraint regulation during exploration.
Moreover, PID-based Lagrangian regulation achieves faster and more stable safety-cost convergence than the gradient-based alternative.
PID-Lag-SAC w/o MSS shows a clear reward ceiling and substantially larger uncertainty, suggesting that the instantaneous constraint-violation cost lacks sufficient forward-looking safety guidance for efficient policy improvement.
Overall, the proposed safe reinforcement learning architecture provides the best balance among convergence speed, training fluctuation suppression, final reward magnitude, and constraint-cost regulation, demonstrating the complementary benefits of the reachability-based safety cost defined by the motion safety set and PID-based multiplier adaptation.

\begin{table}[t]
\centering
\caption{Core inference metrics of the three methods averaged over ten random seeds.}
\label{tab:comparison}
\footnotesize
\resizebox{\columnwidth}{!}{%
\begin{tabular}{lcccc}
\hline
\rule{0pt}{2.65ex}\textbf{Method} & \textbf{Collision-free rate} $\uparrow$ & \textbf{Goal-reaching rate} $\uparrow$ & $\boldsymbol{V_{h,\gamma}^{\max}} \downarrow$ & $\boldsymbol{|\delta|_{\mathrm{mean}}}$ (rad) $\downarrow$ \\[0.3pt]
\hline
PID-Lag-SAC w/ MSS  & 100\% & 100\% & -0.019 & 0.093 \\
PID-Lag-SAC w/o MSS & 100\% & 70\%  & -0.001 & 0.103 \\
Lag-SAC w/ MSS      & 100\% & 90\%  & -0.017 & 0.115 \\
\hline
\end{tabular}
}
\end{table}

\textbf{Quantitative Comparison.}
Table~\ref{tab:comparison} summarizes representative inference metrics after training.
The collision-free rate is the percentage of runs without geometric overlap between the ego vehicle and the obstacle.
The goal-reaching rate is a stricter task-completion metric: a run is counted as successful only if the vehicle bypasses the obstacle and reaches the post-avoidance target at $X_{\text{target}} = 80~\mathrm{m}$ without leaving the road boundary or exhibiting an excessive heading deviation throughout the complete collision-avoidance maneuver.
As shown in Table~\ref{tab:comparison}, all three methods achieve a 100\% collision-free rate, confirming that each method can avoid direct collision in the evaluated runs.
However, the proposed PID-Lag-SAC w/ MSS is the only method that reaches the post-avoidance target in all seeds, while maintaining a clearly negative maximum reachability value and the smallest mean steering magnitude.
Without MSS, PID-Lag-SAC w/o MSS reaches the goal in only 70\% of the runs, and its maximum reachability value is close to zero, indicating a much smaller forward-looking safety margin despite collision-free execution.
With MSS but without PID-based multiplier regulation, Lag-SAC w/ MSS improves the goal-reaching rate to 90\% and maintains a larger reachability-based safety margin, but requires the largest mean steering magnitude among the three methods.
These results show that the proposed method does not simply increase conservativeness; instead, it achieves the best balance among collision avoidance, complete post-avoidance recovery, reachability-based safety margin, and control smoothness.

Taken together, the training and inference results demonstrate that the reachability-based motion safety set constraint improves proactive safety awareness, while PID-Lagrangian regulation stabilizes constrained policy optimization.
The proposed PID-Lag-SAC w/ MSS therefore provides a more reliable solution for emergency collision avoidance under low-adhesion conditions.

\subsection{Real-Vehicle Collision Avoidance Experiments}
\label{sec:result_real_vehicle}

As shown in Fig.~\ref{fig:test_vehicle}, the trained PID-Lag-SAC w/ MSS policy was deployed on a production-level passenger vehicle through a real-time onboard control architecture.
For comparison, the baseline method was designed following the collision-free emergency planning and control framework in \cite{zhao2023collisionfree}, which generates a collision-free reference trajectory based on surrounding-risk assessment and tracks it through a model predictive controller.

The real-vehicle experiment compares the proposed policy and the baseline method in the obstacle-avoidance scenario designed in the simulation study, with an initial speed of approximately 15 m/s and the electronic stability control (ESC) system disabled.
The tests were conducted on a low-adhesion test site, where the effective road adhesion coefficient was approximately $\mu=0.3$.
Although the test surface was prepared to be as spatially uniform as possible, unavoidable road-surface nonuniformity and tire--road interaction variations remain in the real-vehicle experiments.
The objective is to evaluate whether the proposed reachability-guided safe RL policy can provide a larger unified safety margin, jointly determined by the geometric collision-avoidance constraint and the chassis stability constraint, while producing a more compact avoidance trajectory and faster lateral-state recovery than the baseline method when deployed on the real vehicle.

\begin{figure}[t]
\centering
\includegraphics[width=\columnwidth]{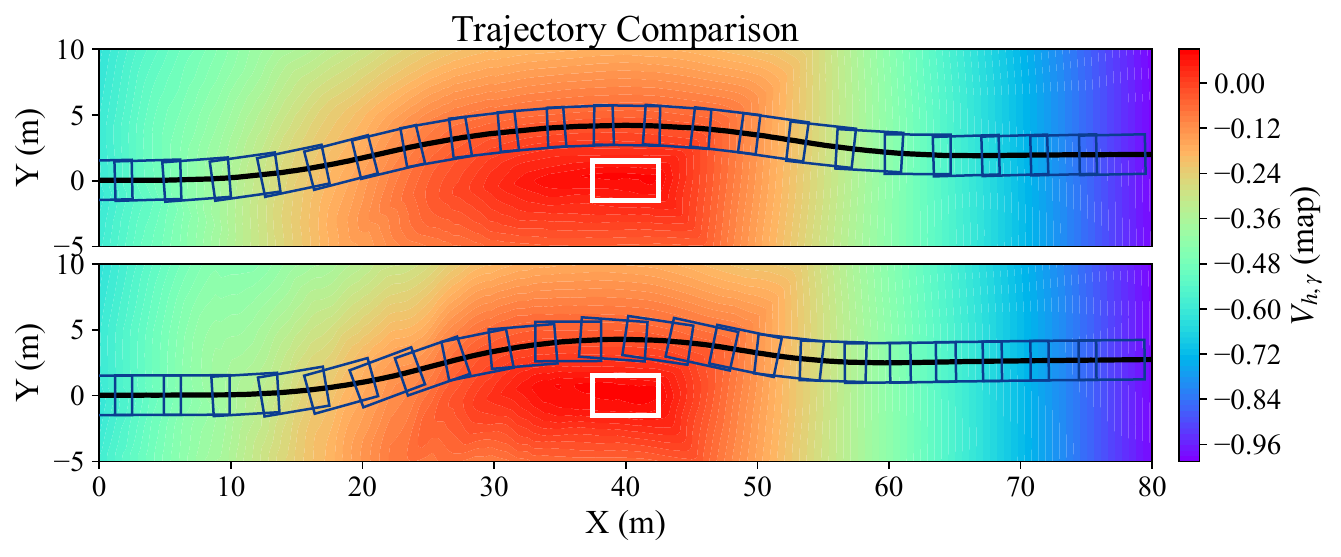}
\caption{Real-vehicle collision avoidance trajectories projected onto the motion safety set. The upper panel shows the proposed method, and the lower panel shows the baseline method.}
\label{fig:real_traj_safety_set}
\end{figure}

\begin{figure}[t]
\centering
\includegraphics[width=0.9\columnwidth]{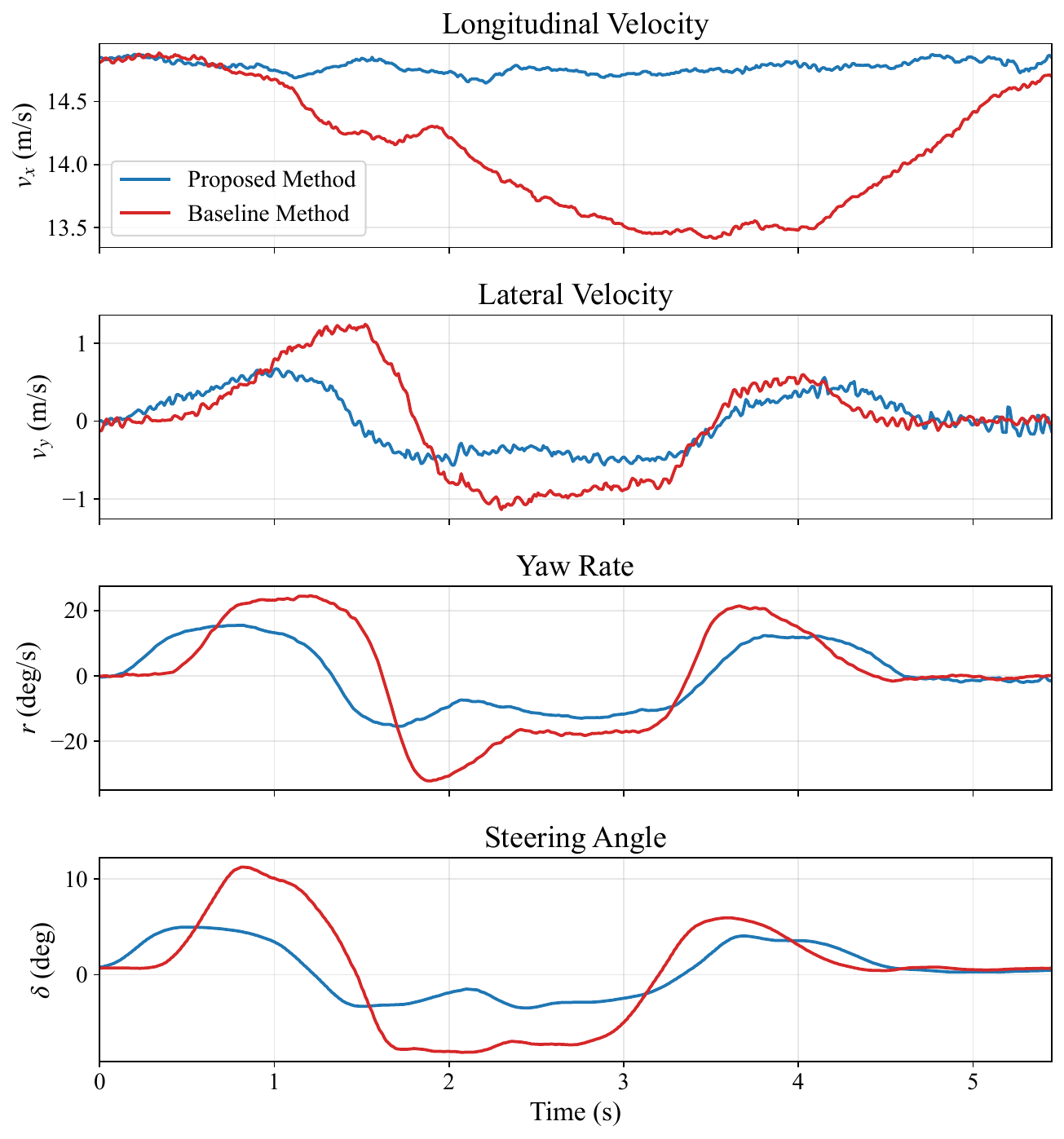}
\caption{Comparison of longitudinal velocity, lateral velocity, yaw rate, and steering response during the real-vehicle collision avoidance experiment.}
\label{fig:vh_vy_r_delta}
\end{figure}

\begin{table}[t]
\centering
\caption{Quantitative comparison of real-vehicle safety metrics. More negative values indicate larger safety margins.}
\label{tab:real_vehicle_metrics}
\scriptsize
\begin{tabular}{lccc}
\toprule
\textbf{Method} & $\boldsymbol{V_{h,\gamma}^{\max}}$ & $\boldsymbol{h_{\mathrm{env}}^{\max}}$ & $\boldsymbol{h_{\mathrm{chassis}}^{\max}}$ \\
\midrule
Proposed method & -0.031 & -1.148 & -3.966 \\
Baseline method & -0.022 & -1.039 & -2.563 \\
\bottomrule
\end{tabular}
\end{table}

Fig.~\ref{fig:real_traj_safety_set} compares the real-vehicle trajectories projected onto the learned reachability value field.
For each longitudinal slice, the value field is evaluated using the measured $v_x$, $v_y$, $r$, and $\psi$, so the displayed safety boundary reflects the local vehicle dynamic state rather than a static obstacle map.
The proposed method follows a more regular avoidance arc and keeps clearer lateral clearance from the obstacle envelope during the closest-approach phase.
Meanwhile, compared with the baseline trajectory that shows a more pronounced post-avoidance sideslip tendency, the proposed trajectory remains in a lower-value region of the reachability field.
The quantitative metrics in Table~\ref{tab:real_vehicle_metrics} further confirm this observation: the proposed method achieves lower maximum reachability values and lower maximum values of both the geometric collision-avoidance constraint and the chassis-stability constraint.
These results show that this comparison case exhibits clearer advantages in both geometric collision-avoidance safety margin and chassis stability margin.

Fig.~\ref{fig:vh_vy_r_delta} further compares the dynamic responses and steering inputs.
The $v_x$ responses are generally comparable, with the speed reduction during avoidance mainly associated with the steering maneuver, which increases lateral tire-force demand and couples longitudinal and lateral vehicle dynamics.
In the lateral direction, the proposed method reduces the peak lateral velocity from 1.24 m/s to 0.67 m/s, showing that the lateral motion is better suppressed during obstacle avoidance and recovery.
The baseline method exhibits a larger lateral velocity variation, which is consistent with the sideslip tendency observed in Fig.~\ref{fig:real_traj_safety_set}.
The yawrate response is also more stable under the proposed method, with the peak yawrate magnitude reduced from 32.1 deg/s to 15.5 deg/s.
In contrast, the baseline yawrate response increases again during the recovery phase, indicating a less settled lateral-yaw motion.
The steering response further confirms that this improvement is not achieved through aggressive steering: the peak absolute steering angle decreases from 11.25 deg to 4.96 deg, and the mean absolute steering rate decreases from 9.10 deg/s to 5.14 deg/s.
The proposed steering input changes more gradually around the avoidance and return phases, whereas the baseline requires a sharper correction after the obstacle is bypassed.
Therefore, the reachability-guided safe RL policy achieves a larger safety margin while producing smoother steering and a more stable vehicle posture.

Overall, the real-vehicle comparison demonstrates that the proposed policy can complete emergency collision avoidance on a low-adhesion road while maintaining a larger reachability-based unified safety margin than the baseline method.
It also yields milder lateral-yaw responses and steering demands, indicating smoother steering behavior and more stable post-avoidance recovery.
These results show that the PID-Lagrangian SAC method with motion-safety-set guidance does not simply lead to a more conservative avoidance maneuver, but provides forward-looking safety guidance that improves both collision-avoidance performance and vehicle dynamic stability.

\section{Conclusion}
\label{sec:conclusion}

This paper presented a motion-safety-set guided safe reinforcement learning framework for emergency collision avoidance under extreme driving conditions. A unified signed safety constraint was first formulated by combining the SAT-based geometric collision margin with a lateral-yaw stability envelope, enabling external obstacle avoidance and internal stability limits to be evaluated within a common formulation. Based on this constraint, a Hamilton-Jacobi reachability-based motion safety set was introduced to characterize a forward-looking safety representation that accounts for future state evolution beyond instantaneous proximity. Furthermore, an offline reinforcement learning scheme based on a discounted reachability Bellman operator was developed to approximate the reachability value function from large-scale extreme-driving data. The learned motion safety set was then incorporated into a PID-Lagrangian SAC framework as a safety cost, providing proactive constraint supervision that anticipates potential safety violations and guides policy optimization toward safe collision-avoidance behaviors.

Offline approximation results based on real-vehicle dataset showed that the learned motion safety set captures the influence of key dynamic factors such as longitudinal velocity and relative yaw angle on collision-avoidance feasibility, confirming that the safety boundary reflects vehicle dynamics beyond pure geometric clearance.
Comparative safe RL training results further demonstrated that the proposed PID-Lagrangian SAC guided by the motion safety set achieves improved learning stability and safety performance compared with baselines using instantaneous unified safety costs or standard Lagrangian updates. 
Real-vehicle experiments on a production passenger vehicle under low-adhesion conditions further verified that the trained policy can complete emergency obstacle avoidance with a smoother trajectory, preserve both geometric collision-avoidance and chassis-stability margins, and reduce lateral displacement while maintaining a coordinated lateral-yaw response.

Future work will extend the proposed framework to more complex emergency driving scenarios, including moving obstacles, multi-agent traffic interactions, and spatially varying road adhesion conditions. In addition, uncertainty-aware reachability learning will be investigated to improve the reliability of learned motion safety boundaries under distribution shifts, unseen extreme maneuvers, and imperfect state estimation.
These extensions will further enhance the safe operational capability of autonomous vehicles in more demanding and unpredictable driving environments.

\section*{Acknowledgements}
\label{sec:acknowledgements}

This work was supported by the National Key Research and Development Program of China under Grant 2024YFB2505104.

\appendix

\section{SAT-Based Collision Detection Details}
\label{app:sat}

Let $j\in\{E,O\}$ denote the ego vehicle and the obstacle, respectively.
The rectangular occupancy of object $j$ is represented by its center position $c_j\in\mathbb{R}^2$, yaw angle $\psi_j$, half-length $l_j$, and half-width $w_j$:
\begin{equation}
\mathcal{R}_j =
\left\{
c_j + R(\psi_j)q
\mid
|q_x|\le l_j,\ |q_y|\le w_j
\right\}.
\label{eq:rect_region}
\end{equation}
where $R(\psi_j)$ is the 2D rotation matrix and $q=[q_x, q_y]^\top$ denotes the local body-frame coordinates.

The longitudinal and lateral unit directions of rectangle $j$ are denoted by $e_j^l$ and $e_j^w$, respectively.
For two rectangles in the plane, the four candidate separating axes are $\mathcal{A} = \{e_E^l,e_E^w,e_O^l,e_O^w\}$.
For an arbitrary axis $a\in\mathcal{A}$, the projection radius of rectangle $j$ is:
\begin{equation}
\rho_j(a) = l_j|a^\top e_j^l| + w_j|a^\top e_j^w|.
\label{eq:projection_radius}
\end{equation}
The signed overlap margin along axis $a$ is then:
\begin{equation}
h_a(\mathbf{x}) = \rho_E(a)+\rho_O(a) - |a^\top(c_E-c_O)|.
\label{eq:axis_margin}
\end{equation}
If $h_a(\mathbf{x})<0$, the two rectangles are separated along axis $a$; if $h_a(\mathbf{x})>0$, their projections overlap.

\section{Lateral-Yaw Stability Envelope Derivation}
\label{app:stability}

The 2-DOF lateral-yaw dynamics, obtained by assuming constant $v_x$, are:
\begin{align}
m v_x (\dot{\beta} + r) &= F_{yf} + F_{yr} \label{eq:bicycle_lat} \\
I_z \dot{r} &= l_f F_{yf} - l_r F_{yr} \label{eq:bicycle_yaw}
\end{align}
Using the Fiala model linearized under unsaturated conditions, the front/rear tire forces are $F_{yf} = -\bar{C}_{\alpha f} \alpha_f$ and $F_{yr} = -\bar{C}_{\alpha r} \alpha_r$.
Substituting yields the state-space system in \eqref{eq:state_space} with:
\begin{equation}
\label{eq:AB_matrix}
\mathbf{A} = \begin{bmatrix}
-\frac{\bar{C}_{\alpha f} + \bar{C}_{\alpha r}}{m v_x} & \frac{l_r \bar{C}_{\alpha r} - l_f \bar{C}_{\alpha f}}{m v_x^2} - 1 \\[6pt]
\frac{l_r \bar{C}_{\alpha r} - l_f \bar{C}_{\alpha f}}{I_z} & -\frac{l_f^2 \bar{C}_{\alpha f} + l_r^2 \bar{C}_{\alpha r}}{I_z v_x}
\end{bmatrix}
\end{equation}

The mapping from slip angles to vehicle-centered variables is:
\begin{equation}
\label{eq:yaw_sideslip_relation}
\left\{
\begin{aligned}
r &= \frac{v_x\left(\alpha_f-\alpha_r+\delta_f\right)}{l_f+l_r} \\[6pt]
\beta &= \frac{l_r\alpha_f+l_f\left(\alpha_r+\delta_f\right)}{l_f+l_r}
\end{aligned}
\right.
\end{equation}

Applying the Routh-Hurwitz criterion ($\text{tr}(\mathbf{A}) < 0$ and $\det(\mathbf{A}) > 0$) leads to the stability boundary condition in \eqref{eq:stability_boundary}.

\section{Proof of the Contraction Property}
\label{app:contraction}

The supremum, infimum, and pointwise maximum operators are 1-Lipschitz under the $\|\cdot\|_\infty$ norm.
Fix any $(\mathbf{x}, \mathbf{u})$. By the definition of $\mathcal{P}^*$:
\begin{multline}
\label{eq:pointwise_diff}
|(\mathcal{P}^* Q_1)(\mathbf{x}, \mathbf{u}) - (\mathcal{P}^* Q_2)(\mathbf{x}, \mathbf{u})| \\
= \gamma \bigl|\sup_{\mathbf{d}} \max\{h(\mathbf{x}), A_1(\mathbf{d})\}
- \sup_{\mathbf{d}} \max\{h(\mathbf{x}), A_2(\mathbf{d})\}\bigr|
\end{multline}
where $\mathbf{x}' = f(\mathbf{x}, \mathbf{u}, \mathbf{d})$ and
$A_i(\mathbf{d}) = \inf_{\mathbf{u}'\in\mathcal{U}} Q_i(\mathbf{x}', \mathbf{u}')$, $i=1,2$.
For each fixed $\mathbf{d}$, applying the 1-Lipschitz property of $\inf$ over $\mathcal{U}$ yields:
\begin{equation}
\label{eq:inf_step}
|A_1(\mathbf{d}) - A_2(\mathbf{d})| \leq \sup_{\mathbf{u}'} |Q_1(\mathbf{x}', \mathbf{u}') - Q_2(\mathbf{x}', \mathbf{u}')|
\end{equation}
Since the map $z \mapsto \max\{h(\mathbf{x}),z\}$ is 1-Lipschitz, combining \eqref{eq:inf_step} with the 1-Lipschitz property of $\sup$ over $\mathcal{D}$ yields:
\begin{align}
\label{eq:sup_step}
\begin{split}
&\bigl|\sup_{\mathbf{d}} \max\{h(\mathbf{x}), A_1(\mathbf{d})\}
- \sup_{\mathbf{d}} \max\{h(\mathbf{x}), A_2(\mathbf{d})\}\bigr| \\
&\quad \leq \sup_{\mathbf{d}} \sup_{\mathbf{u}'}
|Q_1(\mathbf{x}',\mathbf{u}') - Q_2(\mathbf{x}',\mathbf{u}')| \\
&\quad \leq \|Q_1 - Q_2\|_\infty
\end{split}
\end{align}
Combining \eqref{eq:pointwise_diff} and \eqref{eq:sup_step} and taking the supremum over $(\mathbf{x}, \mathbf{u})$ yields \eqref{eq:contraction}. \hfill $\blacksquare$

Since $\mathcal{P}^*$ is a $\gamma$-contraction mapping, the Banach fixed-point theorem guarantees the existence and uniqueness of the fixed point $Q_{h,\gamma}^* = \mathcal{P}^* Q_{h,\gamma}^*$, and the value iteration $Q^{(k+1)} = \mathcal{P}^* Q^{(k)}$ converges globally at a geometric rate $\|Q^{(k)} - Q_{h,\gamma}^*\|_\infty \leq \gamma^k \|Q^{(0)} - Q_{h,\gamma}^*\|_\infty$.
In offline settings with empirical operator $\hat{\mathcal{P}}^*$ satisfying $\|\hat{\mathcal{P}}^* Q - \mathcal{P}^* Q\|_\infty \leq \epsilon$, the fixed-point error is bounded by $\|\hat{Q}_{h,\gamma}^* - Q_{h,\gamma}^*\|_\infty \leq \epsilon / (1 - \gamma)$.

\bibliographystyle{elsarticle-num}
\bibliography{references}

\end{document}